\documentclass[10pt,a4paper,oneside]{article}
\usepackage{enumitem}
\usepackage{multirow}
\usepackage[margin=1in]{geometry}
\newcommand{\newmid}{\ \middle|\ }
\usepackage[utf8]{inputenc}
\usepackage[T1]{fontenc}
\usepackage{mathrsfs, amsmath}
\usepackage{amsfonts}
\usepackage{amssymb}
\usepackage{ifthen}
\usepackage[margin=1in]{geometry}
\usepackage{graphicx}
\usepackage{float}
\usepackage[hidelinks]{hyperref}
\graphicspath{ {figures/} }
\usepackage{array}
 \usepackage{amsmath, xspace}
\usepackage{epsfig}
\usepackage{longtable}
\usepackage{color}
\usepackage{subfigure}

\def\sthat{\ \ \mbox{such that}\ \ }
\usepackage{courier}

\newtheorem{theorem}{Theorem}
\newtheorem{lemma}[theorem]{Lemma}

\def\bkE{{\rm I\kern-.17em E}}
\def\bk1{{\rm 1\kern-.17em l}}
\def\bkD{{\rm I\kern-.17em D}}
\def\bkR{{\rm I\kern-.17em R}}
\def\bkP{{\rm I\kern-.17em P}}

\def\bkZ{{\bf{Z}}}

\def\bfone{{\bf 1}}

\def\bkE{{\rm I\kern-.17em E}}
\def\bk1{{\rm 1\kern-.17em l}}
\def\bkD{{\rm I\kern-.17em D}}
\def\bkR{{\rm I\kern-.17em R}}
\def\bkP{{\rm I\kern-.17em P}}

\def\bkZ{{\bf{Z}}}
\def\b12{(\beta_1,\beta_2)}

\newenvironment{proof}[1][]{{\noindent \bf Proof #1: }}{\hfill \qed\vspace{6pt}}

\newcounter{example}
\renewcommand{\theexample}{\arabic{example}}
\newenvironment{examplec}[1][]{\refstepcounter{example}
\par\medskip \noindent%
   \textbf{Example~\theexample. #1} \rmfamily}{\hfill $\square$   \hspace{-4.5pt} \vspace{6pt}}

\newlength{\noteWidth}
\long\def\notes#1{\ifinner
{\tiny #1}
\else
\marginpar{\parbox[t]{\noteWidth}{\raggedright\tiny #1}}
\fi\typeout{#1}}

 \def\notes#1{\typeout{read notes: #1}} 

\newcommand{\ie}{i.e.\@\xspace} 
\newcommand{\eg}{e.g.\@\xspace} 

\newcommand{\Real}{\mathbb{R}}

\def\In{\mathop{\hbox{\it In}\,}}

\def\half  {{\textstyle{1\over 2}}}

\def\spose#1{\hbox to 0pt{#1\hss}}

\def\text #1{\hbox{\quad#1\quad}}

\def\nthinsp{\mskip -2   mu}

\def\superstar{^{\raise 0.5pt\hbox{$\nthinsp *$}}}
\def\SUPERSTAR{^{\raise 0.5pt\hbox{$*$}}}

\def\lamstarT {\lambda^{\raise 0.5pt\hbox{$\nthinsp *$}T}}

\def\Cstar{C\SUPERSTAR}

\def\Rscr{{\cal R}}
\def\Gscr{\mathcal{G}}

\def\aur{\;\textrm{and}\;}

\def\eef{\;\textrm{if}\;}

\def\LCP{{\rm LCP}}

\let\forallnew\forall
\renewcommand{\forall}{\forallnew\ }
\let\forall\forallnew

\def\t{^\top}
		\def\bkE{{\rm I\kern-.17em E}}
		\def\bk1{{\rm 1\kern-.17em l}}
		\def\bkD{{\rm I\kern-.17em D}}
		\def\bkR{{\rm I\kern-.17em R}}
		\def\bkP{{\rm I\kern-.17em P}}
		\def\bkY{{\bf \kern-.17em Y}}
		\def\bkZ{{\bf \kern-.17em Z}}
		\def\bkC{{\bf  \kern-.17em C}}

{\begin{list}{}%
         {\setlength{\leftmargin}{#1}}%
         \item[]%
}
{\end{list}}

		\def\bsp{\begin{split}}
		\def\beq{\begin{eqnarray}}
		\def\bal{\begin{align*}}
		\def\bc{\begin{center}}
		\def\be{\begin{enumerate}}
		\def\bi{\begin{itemize}}
		\def\bs{\begin{small}}
		\def\bS{\begin{slide}}
		\def\ec{\end{center}}
		\def\ee{\end{enumerate}}
		\def\ei{\end{itemize}}
		\def\es{\end{small}}
		\def\eS{\end{slide}}
		\def\eeq{\end{eqnarray}}
		\def\eal{\end{align*}}
		\def\esp{\end{split}}
		\def\qed{ \vrule height7.5pt width7.5pt depth0pt}  

				\def\viproblem#1#2#3{\fbox
		 {\begin{tabular*}{0.95\textwidth}
			{@{}l@{\extracolsep{\fill}}l@{\extracolsep{6pt}}l@{\extracolsep{\fill}}c@{}}
				#1 &#2 & $#3 $ 
			\end{tabular*}}}
	\def\cp2problem#1#2#3#4{\fbox
		 {\begin{tabular*}{0.9\textwidth}
			{@{}l@{\extracolsep{\fill}}l@{\extracolsep{6pt}}l@{\extracolsep{\fill}}c@{}}
				#1 & & $#4 $ 
			\end{tabular*}}}

\newcommand{\pmat}[1]{\begin{pmatrix} #1 \end{pmatrix}}
		
		\renewcommand{\emph}[1]{\textbf{#1}}

		\def\bkE{{\rm I\kern-.17em E}}
		\def\bk1{{\rm 1\kern-.17em l}}
		\def\bkD{{\rm I\kern-.17em D}}
		\def\bkR{{\rm I\kern-.17em R}}
		\def\bkP{{\rm I\kern-.17em P}}
		
		\def\bkZ{{\bf{Z}}}

\newcommand {\beeq}[1]{\begin{equation}\label{#1}}
\newcommand {\eeeq}{\end{equation}}
\newcommand {\bea}{\begin{eqnarray}}
\newcommand {\eea}{\end{eqnarray}}

\def\texitem#1{\par\smallskip\noindent\hangindent 25pt
               \hbox to 25pt {\hss #1 ~}\ignorespaces}

\title{\bf Refinement of the Equilibrium of Public Goods Games over Networks: Efficiency and Effort of Specialized Equilibria}
\author{Parthe Pandit \and Ankur A. Kulkarni\thanks{Parthe and Ankur are with the Systems and Control Engineering group at the  Indian Institute of Technology Bombay, Mumbai, India 400076. They can be contacted at \texttt{\small parthe.pandit@iitb.ac.in} and \texttt{\small kulkarni.ankur@iitb.ac.in}, respectively.} }

\newcommand{\scn}[2]{\mathcal{E}_{#1}(#2)}
\newcommand{\bfe}{\textbf{e}}
\newcommand{\supp}{{\rm supp}}
\newcommand{\SNE}{{\rm SNE}}
\newcommand{\DNE}{{\rm DNE}}
\newcommand{\NE}{{\rm NE}}
\renewcommand{\Cstar}{\mathsf{C}^*}
\newcommand{\EWstar}{\mathsf{E}_w^*}
\newcommand{\CSstar}{\mathsf{C}^{\mathsf{S}*}}
\newcommand{\CDstar}{\mathsf{C}^{\mathsf{D}*}}
\newcommand{\EWSstar}{\mathsf{E}_w^{\mathsf{S}*}}
\newcommand{\WU}{\mathsf{W}_\mathsf{U}}

\newcommand{\WUSstar}{\mathsf{W}_\mathsf{U}^{\mathsf{S}*}}
\newcommand{\WUDstar}{\mathsf{W}_\mathsf{U}^{\mathsf{D}*}}
\newcommand{\solves}{{\rm\ solves\ }}

\begin{document}
\date{}

\maketitle
\section*{Abstract}

Recently Bramoulle and Kranton~\cite{bramoulle2007public} presented a model for the provision of public goods over a network and showed the existence of a class of Nash equilibria called specialized equilibria wherein some agents exert maximum effort while other agents free ride. We examine the efficiency, effort and cost of specialized equilibria in comparison to other equilibria. Our main results show that the welfare of a particular specialized equilibrium approaches the maximum welfare  amongst all equilibria as the concavity of the benefit function tends to unity. For forest networks a similar result also holds as the concavity approaches zero. Moreover, without any such concavity conditions, there exists for any network  a specialized equilibrium that requires the maximum weighted effort amongst all equilibria. When the network is a forest, a specialized equilibrium also incurs the minimum total cost amongst all equilibria. For well-covered forest networks we show that all welfare maximizing equilibria are specialized and all equilibria incur the same total cost.  Thus we argue that specialized equilibria may be considered as a refinement of the equilibrium of the public goods game. We show several results on the structure and efficiency of equilibria that highlight the role of dependants in the network. 

\section{Introduction}

Recently, Bramoulle and Kranton \cite{bramoulle2007public} introduced a model for studying the incentives of agents for the provision of non-excludable public goods in the presence of a network amongst the agents. In their model, agents benefit from their own effort and also from the collective effort of their neighbours in the network, according to a monotone concave benefit function. Since effort is costly, agents choose their effort levels  based on the efforts by their neighbours. In fact, some agents may choose to free ride, \ie, exert no effort, if the cumulative effort of their neighbours is such that marginal cost of their effort is higher than its marginal benefit. At a Nash equilibrium \cite{nash50equilibrium} of the resulting game, no agent has an incentive to unilaterally deviate from its effort level. Such a model, as the authors describe, is suitable to study the incentives for companies to invest in innovation and research, in the presence of a network of companies.

Bramoulle and Kranton \cite{bramoulle2007public} considered three classes of equilibria. 
In a \textit{specialized equilibrium}, some agents exert maximum effort~(effort they would exert in the absence of a network) and all other agents free ride.
Equilibria where all agents exert a positive effort are called \textit{distributed equilibria} and equilibria which are neither specialized nor distributed are referred to as \textit{hybrid equilibria}. Bramoulle and Kranton showed that there may not exist an efficient equilibrium, \ie, a profile of efforts in equilibrium may not correspond to that which maximizes social welfare. Hence it is relevant to compare only equilibrium profiles based on their welfare and ask which equilibria yield maximum welfare. 

Recall that for a class of games $\mathcal{G}$ the Nash equilibrium may be thought of as a set-valued function that prescribes a set of strategy profiles, $\NE(G)$, for each game in $G\in \mathcal{G}$. A refinement of the Nash equilibrium is another set-valued function which prescribes for each $G\in \mathcal{G}$ a \textit{subset} of strategy profiles, $\Rscr(G) \subseteq \NE(G)$, with the property that, for any $G\in \mathcal{G}$, $\Rscr(G)$ is nonempty if $\NE(G)$ is nonempty, and that there is some $G\in \Gscr$ such that $\Rscr(G) \neq \NE(G)$. Refined equilibria have additional properties specified by the mapping $\Rscr$ and may be regarded as being more attractive for consideration as a solution concept.  Bramoulle and Kranton~\cite{bramoulle2007public} showed that specialized equilibria always exist, and under certain conditions, they have the property of stability under best response dynamics and hence may be considered as a refinement of the Nash equilibrium of the  public goods game.
However, it is not clear how these equilibria rank under the criterion of maximum equilibrium welfare and  whether they lead to maximum total equilibrium effort or minimum total equilibrium cost. More generally, it is of interest to understand how the nature of the network affects the structure and efficiency of equilibria of a public goods game.

The present paper is born out of the observation that answers to these questions  are within reach from our previous work in graph theory \cite{pandit2016linear}. 
It was shown in \cite{bramoulle2007public} that specialized equilibria are in a one-to-one correspondence with maximal independent sets of the network. 
While in  \cite{pandit2016linear}, we provided new characterizations of the cardinality of the \textit{largest} and the \textit{smallest} maximal independent set in a graph.
Building on this characterization, in this paper, we show that specialized equilibria may be considered as a refinement of the equilibrium of a public goods game under certain broad assumptions and criteria. 

Our first result shows that under certain assumptions on the concavity of the benefit function, there exist specialized equilibria which attain the maximum welfare amongst all equilibria. Specifically, there is a particular specialized equilibrium with the property that as the concavity approaches unity, the welfare under this equilibrium comes arbitrarily close to the maximum equilibrium welfare. If the graph is a forest, a similar result holds even as the concavity approaches zero. Moreover, for a class of networks called \textit{well-covered forests} \cite{plummer1993well}, we show that for any benefit function, all welfare maximizing equilibria are necessarily specialized. A well-covered forest is a graph without cycles and isolated vertices wherein every vertex is either adjacent to exactly one other vertex, or has exactly one neighbour that possesses this property. 

Next, considering the total weighted \textit{effort} as the criterion for comparing equilibria, we show that for any benefit function, there exist specialized equilibria that maximize the total weighted effort amongst all equilibria. Hence specialized equilibria may be considered as a refinement of the equilibrium of a public goods game by the criterion of total weighted effort, and also by the criterion of welfare under certain assumptions on the concavity of the benefit function.  Surprisingly, analogous results do not hold for the criterion of minimum total cost, or minimum total effort  -- in general hybrid or distributed equilibria may lead to the least total cost. In fact, in regular networks  distributed equilibria form a refinement of the equilibrium under the criterion of minimum total cost. However, if the network is a forest, we find that there is a specialized equilibrium that attains the minimum total equilibrium cost. Once again, well-covered forest networks have an interesting property -- \textit{all} equilibria on such networks require the \textit{same} total effort. 

Additionally,  we derive results relating the nature of equilibria and their efficiency to the structure of the underlying network of agents. We have found that the absence of cycles (\ie, forest networks) and the presence of dependants, \ie, agents having only one neighbour, is an important characteristic of the network in this regard. These results may have sociological interpretations and may be of independent interest.

The following section introduces the model and formally states the main results of this paper.

\subsection{Model, terminology and main results}

Let the graph $G=(V,E)$  be a network with agents represented by vertices $V=\{1,\hdots,n\}$,  and with links $E$ representing pairwise connectivity between the agents, be it geographical, economic or social \cite{bramoulle2007public}. We first recall some terminology pertaining to graphs. Two vertices $i, j \in V$ are said to be \textit{adjacent} if there exists an link $(i, j) \in E$ between them. Adjacent agents are also called \textit{neighbours}. An independent set of a graph is a set of pairwise nonadjacent vertices. An independent set is said to be \textit{maximal} if it is not a subset of a larger independent set. An independent set with the largest number of vertices is called a \textit{maximum} independent set. The cardinality of a maximum independent set is called the \textit{independence number} and is denoted by $\alpha(G)$. It is easy to see that a maximum independent set is maximal, but the converse is not true. A generalization of the independence number is the $w$-weighted independence number for a vector\footnote{All vectors are column vectors unless mentioned otherwise.}  of weights $w=(w_1,\hdots,w_n)\t$,  $w_i \geq 0$ for all $i \in V$. Define the weight of a set $S \subseteq V$ as $\sum_{i\in S}w_i$. The $w$-weighted independence number $\alpha_w(G)$ is the maximum weight among the weights of all independent sets, and the independent set with the maximum weight is called the $w$-weighted maximum independent set. Clearly $\alpha_{\bfe}(G)$ is the independence number of $G$, where $\bfe := (1,1,\ldots,1)\t$ is the vector of all 1's.

A path between vertices $i$ and $j$ is a collection of distinct links $(y_1,z_1),\hdots,(y_k,z_k)$ such that $y_1=i,z_k=j$ and $z_t=y_{t+1}$ for all $1\leq t< k.$ 
Recall that a graph is said to be \textit{connected} if there is a path between every pair vertices. A path that starts and ends at the same vertex is called a cycle. A forest is a graph without cycles. A connected forest is called a tree.

We now come back to the model for the game. Let every agent $i \in V$ exert an effort $x_i \geq 0$ with a constant marginal cost $c$. We call the vector $x = (x_1, x_2, \ldots, x_n) \t$ as the \textit{profile of efforts}. Let $b: \Real \rightarrow \Real$ be a differentiable concave monotone benefit function, i.e., $b'(y) > 0$, $b''(y) < 0$ for all $y \in \Real,\ y \geq 0$. Moreover let  $b'(e^*) = c$, for some $e^* > 0$. In the absence of the network, an agent benefits $b(x_i)$ on exerting an effort $x_i$, which costs $c x_i$, whereby the utility of agent $i$ is $U_i^{\rm solo} =  b(x_i) - c x_i$. This utility is maximum for a unique effort level $e^*$, since $b$ is an increasing concave function, which by definition, satisfies $\frac{\partial U_i^{\rm solo}}{\partial x_i}(e^*) = 0.$

Due to presence of the network, however, an agent $i$ benefits from the collective efforts exerted by itself and its neighbours $N_i := \{j \mid (i,j) \in E\}$. Hence the benefit of agent $i$ is, $$b\left(x_i + \sum_{j \in N_i} x_j\right)=b(\scn{i}{x}),$$ where, 
\begin{equation}
\scn{i}{x}:=x_i + \sum_{j \in N_i}x_j, \label{eq:effortdef} 
\end{equation}
is the \textit{effort of the closed neighbourhood} of agent $i$. It is the cumulative effort from which agent $i$ can benefit. 
Hence the utility of agent $i$ is \begin{equation}U_i(x) = b(\scn{i}{x})-cx_i.\label{eq:utility}\end{equation}
A profile of efforts $x^*$ is a \textit{Nash equilibrium} \cite{nash50equilibrium} if no agent has an incentive to deviate from it unilaterally, i.e., $$U_i(x^*_i,x^*_{(-i)}) \geq U_i(x_i,x^*_{(-i)}),\quad {\rm for\ all}\ x_i \geq 0,$$ where $x^*_{(-i)} = (x^*_1,\ldots,x^*_{i-1},x^*_{i+1},\ldots,x^*_n) \t$. For a public goods game over a network $G$, we denote by $\NE(G)$ the set of Nash equilibrium effort profiles.

It was shown in \cite{bramoulle2007public}, that for every Nash equilibrium $x$, we have $0 \leq x_i \leq e^*$. 
An agent $i$ is called a \textit{free rider} if $x_i = 0$ and a \textit{specialist} if $x_i = e^*$. 
An equilibrium $x^*$ is said to be \textit{specialized} if $x^*_i=0$ or $x^*_i =e^*$ for all agents $i$. $x^*$ is said to be \textit{distributed} if $x_i^* > 0$ for all agents $i$. 
We denote the set of specialized equilibria by $\SNE(G)$ and the set of distributed equilibria by $\DNE(G)$. Bramoulle and Kranton~\cite{bramoulle2007public} also showed that in any specialized equilibrium, the agents exerting effort $e^*$ form a maximal independent set of the graph. Conversely, every maximal independent set corresponds to a specialized equilibrium in which the agents in the set exert effort $e^*$ and all other agents free ride.  

Among all equilibria of a public goods game, some equilibria may be more efficient than others. 
Carrying forward the discussion by Bramoulle and Kranton, we use the utilitarian definition of welfare to compare efficiency of effort profiles. The \textit{utilitarian} welfare, $\WU(x)$, is the sum of utilities of each agent at the profile of efforts $x=(x_1,x_2,\ldots,x_n)\t$, \begin{equation}\WU(x) := \sum_{i \in V}U_i(x) = \sum_{i \in V}b(\scn{i}{x}) - c\ \bfe \t x. \label{eq:welfare function}\end{equation}
An effort profile $x$ is said to be efficient if $\WU(x) \geq \WU(y)$, for all other effort profiles $y$. It was shown in \cite{bramoulle2007public} that Nash equilibria are in general not efficient.

We denote the \textit{maximum equilibrium welfare} by $\WU^*$. Similarly, for a vector of weights $w\in \Real^n$, $\EWstar$ denotes the {\it maximum $w$-weighted equilibrium effort} and $\Cstar$ the {\it minimum equilibrium cost}. Correspondingly $\WUSstar$, $\EWSstar$ $\CSstar$ are these optimum quantities if only specialized equilibria are considered. Formally,
\begin{align}
\WU^* &= \max \{\WU(x) \mid x \in \NE(G)\}, & \WUSstar &= \max \{\WU(x) \mid x \in \SNE(G)\},\label{eq:maximize utilitarian welfare}\\
\EWstar &= \max \{ w \t x \mid x \in \NE(G)\}, & \EWSstar &= \max \{ w \t x \mid x \in \SNE(G)\},\label{eq:maximize effort}\\
\Cstar &= c \cdot \min \{\bfe \t x \mid x \in \NE(G)\}, & \CSstar &= c \cdot \min \{\bfe \t x \mid x \in \SNE(G)\}.\label{eq:minimize cost}
\end{align}
The set of Nash equilibria which yield the maximum equilibrium welfare $\WU^*$ are called the \textit{welfare maximizing equilibria}, equilibria maximizing the $w$-weighted effort are called {\it effort maximizing equilibria} and equilibria minimizing the cost are called {\it cost minimizing equilibria}.

Clearly, from equations (\ref{eq:maximize utilitarian welfare} - \ref{eq:minimize cost}), we have that,
\begin{equation}\label{eq:basic NE SNE inequalities}
\WU^* \geq \WUSstar,\quad \EWstar \geq \EWSstar \quad \aur \quad \Cstar \leq \CSstar.
\end{equation}
The problem of showing that specialized equilibria are a refinement insofar as maximum welfare is concerned amounts to showing the equality $\WU^* = \WUSstar$. Similarly, to show that specialized equilibria are a refinement with respect to maximum total weighted effort and minimum total cost we need to show that $\EWstar = \EWSstar$ and $\Cstar = \CSstar$, respectively.

Notice that $\scn{i}{\cdot}$ is a linear function of $x$ and hence $b(\scn{i}{\cdot})$ is a concave function, whereby the objective functions in (\ref{eq:maximize utilitarian welfare} - \ref{eq:minimize cost}) are either concave or linear. However the feasible region $\NE(G)$ is non-convex and non-polyhedral and $\SNE(G)$ is an integer lattice (this follows from Theorem \ref{lem:NE and LCP} later). Hence, showing that specialized equilibria are a refinement under any of the criteria above is a hard problem. Importantly, the non-polyhedrality of $\NE(G)$ implies that these claims cannot be arrived at using an argument based on extreme points or linear programming.

Let $$\sigma_b := \frac{b(n e^*) - b(e^*)}{c e ^* (n - 1)},$$ denote the \textit{concavity} of the benefit function $b$ as defined by \cite{bramoulle2007public}. It is the normalized slope of the ``secant'' between $e^*$ and $ne^*$ for the function $b$. Since $b$ is a strictly increasing function, $\sigma_b > 0$, whereas since $b$ is strictly concave, the slope of the secant is less than the slope of the tangent at $e^*$, whereby $c \sigma_b < b'(e^*)=c$, \ie, $\sigma_b < 1$. 

In our analysis we consider two extreme cases, namely, as $\sigma_b$ approaches unity and as $\sigma_b$ approaches zero. While we vary $\sigma_b,$ we assume that $b'(e^*) = c$ continues to hold and $b(e^*)$ remains unchanged with $\sigma_b$.  The two cases are depicted in Figure \ref{fig:concavity extremes}. As $\sigma_b \rightarrow 1$, the slope of the secant tends to the slope of $b$ at $e^*$. Since the function lies between the tangent and the secant, the function is \textit{near-linear} between $e^*$ and $ne^*$ as $\sigma_b \rightarrow 1$. On the other hand, as $\sigma_b \rightarrow 0$, the benefit function \textit{plateaus} beyond $e^*$. 

\begin{figure}
\center
\includegraphics[width=0.6\textwidth]{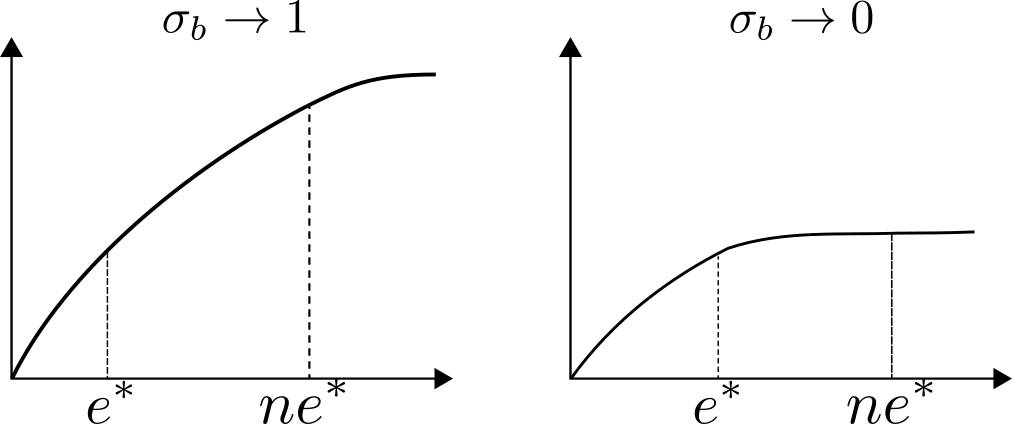}
\caption{ \label{fig:concavity extremes} This figure (not to scale) describes the two extreme cases of concavity. Observe that as $\sigma_b \rightarrow 1$, the benefit function is nearly linear in the interval $[e^*,ne^*]$ and having slope $b'(e^*)$, whereas as $\sigma_b \rightarrow 0$, the benefit function plateaus. Note however that as $\sigma_b$ is varied, $b(e^*)$ and $b'(e^*)$ remains fixed whereby there is no scaling and hence equilibria remain unchanged.}
\end{figure}

An important consequence of holding $b(e^*)$ and $b'(e^*)$ fixed is that it ensures that as $\sigma_b$ changes, the Nash equilibria of the game do not change (see Lemma~\ref{lem:NE and LCP}). This allows one to study the welfare of particular equilibria  in comparison with $\WU^*$ with varying $\sigma_b$.
Bramoulle and Kranton showed that as $\sigma_b \rightarrow 1$ the welfare function approaches the total $d$-weighted effort (plus a constant), where $d$ is a vector of degrees\footnote{The degree of a vertex in a graph is the number of neighbours of the vertex.} of agents. We show that as $\sigma_b \rightarrow 0$ the welfare function equals  the \textit{negative} of the total cost (plus a constant).  As a consequence, when $\sigma_b \rightarrow 1$, the problem of equilibrium  welfare maximization resembles the maximization of the total $d$-weighted equilibrium effort while as $\sigma_b \rightarrow 0$, it resembles the total equilibrium cost minimization.

Our first main result in this paper concerns these extreme regimes of $\sigma_b$.

\begin{theorem}
\label{thm:main}{\bf (Welfare of specialized equilibria)}\\
For a public goods game over a network without isolated agents,
\begin{enumerate}[label=(\alph*)]
\item \label{thm:main1} Let $d=(d_1,\hdots,d_n)\t$ be the vector of degrees of agents, \ie, $d_i=|N_i|$ for all $i\in V.$ Let $S^*$ be a maximum $d$-weighted independent set in the network and let 
$x^*$ be the specialized equilibrium with agents in $S^*$ as specialists and the rest as free riders. Then, as $\sigma_b \rightarrow 1$, the welfare of $x^*$ 
approaches the maximum equilibrium welfare, \ie, $$\lim_{\sigma_b \rightarrow 1}\WU^* = \lim_{\sigma_b \rightarrow 1} \WU(x^*)=n\big(b(e^*)-ce^*\big) + ce^*\alpha_d(G).$$
Consequently, $\lim_{\sigma_b \rightarrow 1} \WU^*=\lim_{\sigma_b \rightarrow 1} \WUSstar$.
\item \label{thm:main2} Suppose the network is a forest. Let $S'$ be the smallest maximal independent set in the network and let 
$x'$ be the specialized equilibrium with agents in $S'$ as specialists and the rest as free riders. Then as $\sigma_b \rightarrow 0$, the welfare of $x'$ 
approaches the maximum equilibrium welfare, 
\ie, $$\lim_{\sigma_b \rightarrow 0}\WU^* = \lim_{\sigma_b \rightarrow 0}\WU(x')=nb(e^*)-ce^*\beta(G),$$
where $\beta(G)$ is the size of $S'.$ 
 Hence, for forests, $\lim_{\sigma_b \rightarrow 0} \WU^*=\lim_{\sigma_b \rightarrow 0} \WUSstar$. 
\item \label{thm:main3} The welfare of any distributed equilibrium (if it exists) is independent of $\sigma_b$. Moreover, $$\WUDstar \geq \lim_{\sigma_b \rightarrow 0}\WUSstar,$$ where, $$\WUDstar =\inf\{\WU(x) \mid x \in \DNE(G)\}, $$ is the least welfare attained by a distributed equilibrium.

\end{enumerate}
\end{theorem}

The assumption regarding the absence of isolated agents does not cause loss of generality, and has been made here for the sake of simplicity. Isolated agents $i$ derive benefits from their own efforts only and hence exert the same effort $e^*$ in any equilibrium.   Hence their utility at an equilibrium $x$ is always $U_i^{\rm solo}(x) = b(e^*) - ce^*$ which is independent of $\sigma_b$ (since $b(e^*)$ and $b'(e^*)$ are fixed). 

Notice that Theorem~\ref{thm:main}\ref{thm:main1} is a general result pertaining to any network. Equivalently, it says that for any network $G$ and $\epsilon>0$, there exists a threshold $\sigma_H$  such that for all benefit functions with fixed $b(e^*)$ and $b'(e^*)=c$ and $\sigma_b \geq \sigma_H$, 
\[0\leq \WU^*-\WU(x^*) \leq \epsilon.\]
 Theorem~\ref{thm:main}\ref{thm:main2} pertains specifically to forests (Theorem~\ref{thm:main}\ref{thm:main1} applies to forests too). Here we have that for any forest $G$ and $\epsilon>0$, there exists a threshold $\sigma_L$ such that for all benefit functions with fixed $b(e^*)$ and $b'(e^*)=c$ and $\sigma_b \leq \sigma_L$,
 \[0 \leq \WU^* - \WU(x') \leq \epsilon.\] 
 In Theorem~\ref{thm:main}\ref{thm:main3}, notice that $\WUDstar$ is the \textit{least} welfare over distributed equilibria. 
The above theorem is a consequence of a more general result regarding the ranking of equilibria when they are compared based on their total weighted effort and total cost which are given by Theorem \ref{thm:total effort} below.

It is worth noting that in both Theorems \ref{thm:main} and \ref{thm:total effort}, we not only show that there exists a specialized equilibrium, but also point to a particular equilibrium which approaches the maximum welfare amongst all equilibria, under the given conditions on the concavity or the structure of the network.

\begin{theorem}\label{thm:total effort} {\bf (Total effort and cost of specialized equilibria)} \\For a public goods game over a network, 
\begin{enumerate}[label = (\alph*)]
\item \label{thm:total effort_max} Let $w=(w_1,\hdots,w_n)\t, $ $w_i \geq 0$ for all $i \in V$, be a set of weights and $S^*$ be a maximum $w$-weighted independent set. Then the specialized equilibrium in which agents in $S^*$ exert effort $e^*$ and the rest are free riders is an effort maximizing equilibrium. Consequently, $\EWstar = \EWSstar$. 
\item \label{thm:total cost_forest} Suppose the network is a forest, and let $S'$ be the smallest maximal independent set in the network. Then the specialized equilibrium in which agents in $S'$ exert effort $e^*$ and the rest are free riders is a cost minimizing equilibrium. Consequently, for forest networks, $\Cstar = \CSstar$.
\item \label{thm:total cost} If there exists a distributed equilibrium $x=(x_1,\hdots,x_n)\t$, the cost incurred in it, $\sum_i c x_i$, is at most as much as that of any specialized equilibrium, whereby $\Cstar \leq \CDstar \leq \CSstar$, where $\CDstar$ the minimum cost attained by a distributed equilibrium.
\item \label{thm:total cost_regular} If the network is regular\footnote{A network is said to be regular if all agents have the same number of neighbours.}, a distributed equilibrium in which all agents exert equal effort minimizes the equilibrium cost, \ie, $\Cstar = \CDstar$.
\end{enumerate}
\end{theorem}
Notice that the above theorem does not require any assumptions about the benefit function. 

We show that a particular class of networks called \textit{well-covered forests} happen to have even stronger and more interesting properties. A network $G$ is said to be \textit{well-covered} if all maximal independent sets of $G$ have the same size. 
For a network, if an agent $i$ is adjacent to exactly one other agent $j$, we say that $i$ is a \textit{dependant} of $j$, and $j$ is the \textit{guardian} of $i$. 
It was shown by Plummer \cite{plummer1993well} that a forest network (without isolated agents) is well-covered if and only if all agents are either dependants or guardians, and every guardian agent has exactly one dependant. 
Our last main result is the following theorem regarding welfare maximizing equilibria in well-covered forest networks. Notice that this result holds without any assumptions on the benefit function or $\sigma_b$.

\begin{theorem}\label{thm:well-covered}{\bf (Efficient specialized equilibria in well-covered networks)}\\ For a public goods game over a well-covered forest network,
\begin{enumerate}[label=(\alph*)]
\item\label{thm:well-covered welfare} An equilibrium yields the maximum equilibrium welfare only if it is a specialized equilibrium.
\item\label{thm:well-covered cost} If the network does not have isolated agents, the cost incurred by any equilibrium profile of efforts is $\half ce^* n$.
\end{enumerate}

\end{theorem}
All our results are comprehensively summarized in Table \ref{tab:summary}. 

Additionally, we obtain results that clarify the structure of equilibrium profiles in such games. In an equilibrium, if two neighbouring agents $i$ and $j$ exert positive efforts $x_i$ and $x_j$ such that $x_i + x_j = e^*$, then we call such agents  \textit{co-specialists}, and the link $(i,j)$ joining them is called a \textit{co-specialist link}.
We show in Section \ref{sec:structure} that, in any network at equilibrium, dependants are either specialists, or free riders, or that they form a co-specialist pair with their guardian. 
Interestingly, for networks with at least one dependant, we find that any Nash equilibrium has at least one free rider, \ie, there exists no distributed equilibrium. 
We further show that if the network is a forest, then at equilibrium, agents are either specialists, co-specialists or free riders only. 

\begin{table}
\begin{tabular}{|p{.21\textwidth}|p{.23\textwidth}|p{.23\textwidth}|p{.23\textwidth}|}
\hline
\multirow{2}{*}{} & \multicolumn{3}{c|}{\large\textit{If network is $\hdots$}}\\ \cline{2-4}\textit{\large Then $\hdots$} & \multicolumn{1}{c|}{Any network} & \multicolumn{1}{c|}{Forest}& \multicolumn{1}{c|}{Well-covered forest} \\ \hline 
Set of welfare maximizing equilibria & contains a specialized equilibrium  as $\sigma_b \rightarrow 1$ & contains a specialized equilibrium as $\sigma_b \rightarrow 1$ and as $\sigma_b \rightarrow 0$ &  contains \textit{only} specialized equilibria \\ \hline
Set of effort maximizing equilibria & contains a specialized equilibrium & contains a specialized equilibrium & contains a specialized equilibrium\\ \hline
Set of cost minimizing equilibria & can't say$^*$ & contains a specialized equilibrium & any equilibrium incurs the same cost\\ \hline
\end{tabular}

\caption{\label{tab:summary} Summary of the main results in this paper -- Theorems \ref{thm:main}, \ref{thm:total effort} and \ref{thm:well-covered}. \newline 
$\ ^*$ If the network is regular then there always exists a distributed equilibrium where all agents exert equal effort. This equilibrium minimizes the total cost.}
\end{table}

\subsection{Organization of the paper}

The rest of the paper is organised as follows. Section 2 provides some background on graphs and main results from our previous work in graph theory. In Section 3, we characterize the Nash equilibria in a public goods game over a network and discuss the effect of the network on the structure of the equilibrium. In particular, we discuss results about networks which have \textit{dependant} and \textit{guardian} agents. The proof of Theorem 2 is given in Section 4 which shows that specialized equilibria are a refinement under the criterion of maximum weighted equilibrium effort for all networks, while they are a refinement under the criterion of minimum equilibrium cost if the network is a forest. Section 5 proves Theorem 1 and shows that the specialized equilibrium is a refinement with respect to maximum equilibrium welfare under conditions on the concavity of the benefit function. In section 6, the result in Theorem 3 is proved which gives conditions on equilibria in a well-covered forest network for any benefit function. Finally, section 7 concludes the paper. In every case where the specialized equilibria form a refinement of the Nash equilibrium, we give the specialized equilibrium which is optimal.

\section{Background}
\subsection{Independent and dominating sets}
The adjacency matrix $A=[a_{ij}]$ of a graph $G$ is a $n \times n$ 0-1 matrix such that $a_{ij} = 1$ if and only if $(i,j) \in E$. 
The \textit{characteristic vector} of a set $S \subseteq V$ is a vector $\bfone_S$ with $|V|$ components such that $(\bfone_{S})_i = 1$ if $i \in S$ and 0 otherwise. $S$ is an independent set of $G$ if and only if $\bfone_S\t A \bfone_S=0.$
A maximal independent set is always a \textit{dominating set}, i.e., a set such that any vertex not in the set has at least one neighbour in the set. The minimum cardinality among maximal independent sets is called the \textit{independent domination number} and is denoted by $\beta(G)$. 
A graph $G$ is said to be \textit{well-covered} if all maximal independent sets are of the same size, \ie, $\alpha(G) = \beta(G)$.

Closely related to independent sets is the concept of \textit{matchings}. A matching is a set of links such that no two links have a vertex in common. If every vertex has an link incident on it from a matching then the matching is called a \textit{perfect matching}.

Given a set $S \subseteq V$, the subgraph of $G$ induced by $S$ is the graph $G_S = (S,E_S)$, where $E_S = \{(i,j)\mid i,j \in S \aur (i,j) \in E\}$. Clearly $G_V = G$. 
Degree of a vertex, defined $d_i:=|N_i|$, is the number of neighbours of vertex $i$. 
For a vector $x \in \Real^{|V|}$, we define the \textit{support} of $x$ as
$${\rm supp}(x) := \{i \in V \mid x_i \neq 0\}.$$
In the context of public goods provision over a network, for a profile of efforts $x$ in a network $G$, the set of agents $\supp(x)$ are called the \textit{supporting agents}, \ie, agents who are not free riders. The graph induced by them, \ie, $G_{\supp(x)}$ is called the \textit{network of supporting agent}s. 

\subsection{The linear complementarity problem and maximal independent sets}
\label{sec:previous work}
In this section we recall our results from~\cite{pandit2016linear}. For this purpose we first recall some concepts pertaining to linear complementarity problems; a definitive resource for this is~\cite{cottle92linear}. Given a matrix $M \in \Real^{n \times n}$ and a vector $q \in \Real^n$,  
the linear complementarity problem \LCP($M,q$) is the following problem,
\begin{align*}
\label{eq:LCP}
\text{Find} x=(x_1,x_2\cdots x_n) \in \Real^n \quad \sthat \quad
& (1) \quad x \geq 0,\\
& (2) \quad y := Mx + q \geq 0, \tag*{LCP($M,q$)} \\
& (3) \quad y \t x = 0. 
\end{align*}
LCPs generalize Nash Equilibria in bimatrix games, quadratic programs and several other problems.
Consider a simultaneous move game with two players having loss matrices $A, B \in \Real^{m\times n}$. A mixed strategy Nash equilibrium~\cite{nash50equilibrium} is a pair of vectors $(x^*,y^*) \in \Delta_n \times \Delta_m$ such that, 
	\[(x^*)\t A y^* \leq x \t A y^*, \quad \forall \;x \in \Delta_n, \qquad (x^*)\t B y^* \leq (x^*)\t B y,  \quad \forall \;y \in \Delta_m,\]
	where $\Delta_k :=\{x \in \Real^k \mid \sum_i x_i =1, x\geq 0\}.$  Under certain technical assumptions~(see, \eg, \cite[p.\ 6]{cottle92linear}), it can be shown that if $(x^*,y^*)$ is a Nash equilibrium, then the concatenated vector $\big[{x'}\t,\ {y'}\t\big]\t$ solves $\LCP(M,q)$, where, 
\[x'=x^*/(x^*)\t By^* \qquad y'=y^*/(x^*)\t A y^*,\]	
	  and, 
\[M = \pmat{0 & A \\ B\t &0}, \qquad q = -\bfe,\]
where $\bfe$ denotes a vector of ones in $\Real^{m+n}.$ Conversely, if $\big[{x'}\t\ {y'}\t\big]\t$ solves  $\LCP(M,q)$ then $x^*= x'/(\sum_i x'_i)$ and $y^*=y'/\sum_j y'_j$ is a Nash equilibrium. In  general, equilibria of certain games involving coupled constraints~\cite{kulkarni09refinement} also reduce to LCPs. For more applications, we refer the reader to~\cite{cottle92linear}.

In our previous work \cite{pandit2016linear}, an \LCP~based continuous optimization formulation was provided for finding the independence number and independent domination number of a graph. Given a graph $G=(V,E)$ with adjacency matrix $A$, consider the problem $\LCP(G)$, 
\[ \viproblem{$\LCP(G)$}{Find $ x\in \Real^n $ such that }{x\geq 0,\ (A+I)x \geq \bfe, \ x \t \big((A+I)x - \bfe\big)= 0,} \]
where $I$ is the $|V| \times |V|$ identity matrix, and $\bfe$ is the vector in $\Real^{|V|}$ with all 1's.
The results from our previous work \cite{pandit2016linear} that are relevant here are summarized as the following theorem.


\begin{theorem} {\bf (LCP characterization for maximal independent sets)}\label{thm:old results}
Consider a graph $G = (V,E)$
\begin{enumerate}[label=(\alph*)]
\item \label{thm:characteristic vectors} \cite[Lemma 4]{pandit2016linear} A binary vector is a solution to  $\LCP(G)$ if and only if it is the characteristic vector of a maximal independent set.
\item \label{thm:wTx} \cite[Theorem 1]{pandit2016linear} For a non-negative vector $w \in \Real^n$,
\[\max \{w \t x \mid x \ {\rm  solves }\ \LCP(G)\}  = \alpha_w(G),\]
and the characteristic vector of the $w$-weighted maximum independent set is a maximizer.
\item \label{thm:beta} \cite[Example 1]{pandit2016linear} For a graph $G$,
\[\min \{\bfe \t x \mid x \ {\rm  solves }\ \LCP(G)\}  \leq \beta(G),\]
and the inequality is in general strict.
\item \label{thm:beta equality} \cite[Theorem 2]{pandit2016linear} If the graph $G$ is a forest,
\[\min \{\bfe \t x \mid x \ {\rm  solves }\ \LCP(G)\}  = \beta(G),\] and the characteristic vector of the smallest maximal independent set is a minimizer.
\end{enumerate}
\end{theorem}

Theorem \ref{thm:old results}\ref{thm:characteristic vectors} shows a relation between maximal independent sets of a graph and the solutions of the $\LCP(G)$. The only integral solutions of $\LCP(G)$ are characteristic vectors of maximal independent sets of $G$. Since the $w$-weighted maximum independent set is also maximal, it is a feasible solution to the maximization problem in Theorem \ref{thm:old results}\ref{thm:wTx}. However Theorem~\ref{thm:old results}\ref{thm:wTx} states that it is in fact a maximizer of this problem.

One might expect that analogous to the result in Theorem \ref{thm:old results}\ref{thm:wTx}, the characteristic vector of the smallest maximal independent set is a minimizer of $\bfe \t x$ among all solutions to $\LCP(G)$. However this is not true in general as given by Theorem \ref{thm:old results}\ref{thm:beta}. The gap is shown to be strict even if the graph is bipartite and regular (see \cite[Example 1]{pandit2016linear}).

The equality in Theorem \ref{thm:old results}\ref{thm:beta} is always attained if the graph $G$ is a forest, as indicated by Theorem \ref{thm:old results}\ref{thm:beta equality}.
This is attributed to the peculiar structure of the solution set of $\LCP(G)$ when $G$ is a forest. 
Other than the above results, we also state here a few lemmas from our previous work that we use in the present paper. We refer the reader to \cite{pandit2016linear} for detailed proofs to these claims.

\begin{lemma}\label{lem:old lemmas} For a graph $G$,
\begin{enumerate}[label=(\alph*)]
\item \label{lem:regular} \cite[Lemma 8]{pandit2016linear} If $G$ has $n$ vertices and is regular with degree $d$, then \begin{equation}
\label{eq:regular} \min \{\bfe \t x \mid x \ {\rm  solves }\ \LCP(G)\}  = \frac{n}{d+1},
\end{equation} and the vector $\frac{\bfe}{d+1}$ is a minimizer.
\item \label{lem:old lemmas forest2} \cite[Lemma 9]{pandit2016linear} If $G$ is a forest, and $x$ is a solution to $\LCP(G)$, then there exists a maximal independent set $S \subseteq \supp(x)$.
\end{enumerate}
\end{lemma}


We now apply the results in Theorem \ref{thm:old results} and Lemma \ref{lem:old lemmas} to give conditions under which specialized equilibria are a refinement of the Nash equilibrium of the public goods games over networks. We do this by establishing a relation between $\NE(G)$ and the solution set of $\LCP(G)$ and thereby obtaining a handle on the structure of equilibria in these games.

\section{Structure of equilibria} \label{sec:structure}

\subsection{Characterization of equilibria in a public goods game over a network}

In this section, we show a relation between LCPs and equilibria of the public goods game and establish a few properties regarding the equilibria. We first recall the conditions given by Bramoulle and Kranton, on the effort levels of the agents at a Nash equilibrium.

\begin{lemma}[Section 3.1 \cite{bramoulle2007public}]
\label{lem:equilibrium basic conditions}{\bf Conditions on efforts at equilibrium}\\
A profile of efforts $x^* \geq 0$ is an equilibrium of a public goods game in a network $G$ if and only if exactly one of the following is true, \begin{enumerate}[label=(\alph*)]
\item \label{lem:condition1}$\sum_{j \in N_i}x^*_j > e^*$, and $x^*_i = 0$,
\item \label{lem:condition2}$\sum_{j \in N_i}x^*_j \leq e^*$, and $x^*_i = e^* - \sum_{j \in N_i}x_j$.
\end{enumerate}
\end{lemma}

The above lemma is argued as follows, an agent $i$ has incentive to exert positive effort only if the total effort from which it benefits, $\sum_{j \in N_i}x_j^*$, is at most as much as the effort level for which marginal benefit equals the marginal cost. If however the total effort of the neighbours of $i$ is lesser than this effort level, then $i$ has incentive to exert effort equal to this deficit, but no more.

As a consequence of Lemma \ref{lem:equilibrium basic conditions}, if $x$ is an equilibrium profile then, \begin{equation} \label{eq:less effort in network} 0 \leq x_i \leq e^*, \end{equation} for all $i \in V$, indicating that presence of a network leads to lower effort by agents, as one might expect. We refer to $\frac{1}{e^*}x$ as the \textit{normalized} profile of efforts.

We observe that the conditions given in Lemma \ref{lem:equilibrium basic conditions} resemble the \textit{either-or} nature given by the equations in \eqref{eq:LCP}. We show in the following theorem that the equilibria of a public goods game over a network $G$ are exactly characterized by $\LCP(G)$.

\begin{theorem}
\label{lem:NE and LCP} {\bf (Normalized equilibrium efforts are solutions to \LCP(G))}\\
A profile of efforts $x$ is a Nash equilibrium of the public goods game over the network $G$ if and only if $$\frac{1}{e^*}x\quad {\rm solves}\quad \LCP(G).$$
\end{theorem}
\begin{proof} A vector $x \in \Real^n$ solves $\LCP(G)$ means,
\begin{align}
&x \geq 0,  &&\qquad \Leftrightarrow 
\qquad x_i \geq 0, \quad\forall \;i\in\; V, \label{eq:lcpg1} \\
&(A_G + I)x \geq \bfe, &&\qquad \Leftrightarrow \qquad \scn{i}{x} \geq 1, \quad\forall \;i\in\; V, \label{eq:lcpg2} \\
&y \t \big((A+I)x - \bfe\big) = 0. &&\qquad \Leftrightarrow \qquad x_i (\scn{i}{x}-1)=0, \quad\forall \;i\in\; V. \label{eq:lcpg3} 
\end{align}

Recall that the conditions given in Lemma \ref{lem:equilibrium basic conditions} are both necessary and sufficient for a profile of efforts $x^*$ to be a Nash equilibrium of a public goods game over a network $G$. Writing these conditions differently we get that, $x^*$ is a Nash equilibrium if and only if exactly one of the following is true.
\begin{enumerate}[label=(\alph*)]
\item $x_i^* = 0$ and $\scn{i}{x^*} > e^*$
\item $\scn{i}{x^*} = e^*$ and $x_i^* \geq 0$
\end{enumerate}
Since $\scn{i}{\cdot}$ is a linear function $\scn{i}{\frac{1}{e^*} x^*} = \frac{1}{e^*} \scn{i}{x^*}$. Hence the above conditions are equivalent to saying that if $x^*$ is an equilibrium, then for each $i$, $\frac{1}{e^*}x^*_i \geq 0$ and $\scn{i}{\frac{1}{e^*}x^*} \geq 1$, and at least one of these inequality is tight, i.e.,  $\frac{1}{e^*}x^*\big(\scn{i}{\frac{1}{e^*}x^*} - 1\big) = 0$.
This proves the theorem.
\end{proof}

Thus at equilibrium, the effort of every agent $i$ obeys the conditions given by (EE$(G)$),
\[\boxed{ {\rm Equilibrium\ Effort } \ \ \qquad\  \qquad  x_i\geq 0,\qquad \scn{i}{x} \geq e^*\qquad \aur \qquad x_i = 0 {\rm \ \ or\ \ } \scn{i}{x} = e^*.\tag{EE(G)}}\label{eq:equilibrium efforts}\]

We now apply our results regarding $\LCP(G)$ to Theorem \ref{lem:NE and LCP} to establish our first result regarding the general structure of equilibria of a public goods game over a network. Recall that if two neighbouring agents $i$ and $j$ exert efforts $x_i>0$ and $x_j>0$ such that $x_i + x_j = e^*$, then such agents are referred to as {co-specialists}, and the link $(i,j)$ joining them is called a {co-specialist link}.

\begin{lemma}
\label{lem:basic results of LCP(G)} In a public goods game over a network $G$, 
\begin{enumerate}[label=(\alph*)]

\item The supporting agents in any equilibrium effort profile form a dominating set of the network.

\item \label{lem:undisturbed}If a free rider leaves the network at equilibrium, then the equilibrium remains undisturbed.

\item \label{lem:specialist IS}In any equilibrium, neighbours of a specialist free ride, whereby the specialists form an independent set of the network.

\item \label{lem:co-specialists matching}In any equilibrium, neighbours of both co-specialists free ride, whereby co-specialist links form a matching of the network.

\item \label{lem:bounds on C_j} If $x$ is an equilibrium profile of efforts and $d_i = |N_i|$ is the degree of the agent $i$ such that $d_i>0$, then
 \begin{equation}x_i \leq e^* \leq \scn{i}{x} \leq d_i e^*.\label{lem:inequalities in efforts}
\end{equation} Equality holds in the second inequality if agent $i$ is not a free rider; further, equality holds in both the first and second inequality only if $i$ is a specialist. The third inequality holds with equality if and only if $i$ is a free rider such that all its neighbours are specialists. If $d_i=0, x_i=\scn{i}{x}=e^*.$
\end{enumerate}
\end{lemma}
\begin{proof}[of Lemma \ref{lem:basic results of LCP(G)}] See Appendix
\end{proof}

\subsection{Networks with dependants}
We now show a few results about the structure of the equilibria of public goods provision over a network containing at least one dependant-guardian pair.
Recall that, in a network, if an agent $i$ is adjacent to only one other agent $j$, then we call $i$ a \textit{dependant} of $j$ ($j$ is called the guardian of $i$).
The link $(i,j)$ linking a dependant to its guardian is called a \textit{pendant line} (a term borrowed from graph theory). If $j$ and $i$ are dependants of each other, we call them co-dependants.
Table \ref{tab:dependants} the summarizes results regarding equilibria in networks which have at least one dependant-guardian pair. In Table~\ref{tab:dependants}, by stable equilibrium we mean stability under best-response dynamics (as defined and considered by~\cite{bramoulle2007public}).

\begin{table}[]
\centering

\begin{tabular}{|p{0.30\textwidth}|p{0.23\textwidth}|p{0.10\textwidth}|p{0.23\textwidth}|}
\hline
& \multicolumn{3}{c|}{\large \it If the equilibrium is \ldots}\\ \cline{2-4} {\large \it Then \ldots}
& \multicolumn{1}{c|}{1. Any  equilibrium} & \multicolumn{1}{c|}{2. Stable} & \multicolumn{1}{c|}{3. Unstable}  \\ \hline
1. an agent which is the single dependant of its guardian  is  & either a specialist, free rider or co-specialist & \multicolumn{1}{c|}{a specialist} & either a specialist, free rider or  co-specialist\\ \hline 
2. agents that are \textit{not} single dependants of their guardian are  & all specialists or all free riders & \multicolumn{1}{c|}{all specialists} & all specialists or all free riders  \\ \hline
\end{tabular}
\caption{\label{tab:dependants}This table describes efforts of dependant agents based on whether an equilibrium is stable or not.}
\end{table}

\begin{proof}[of Table \ref{tab:dependants}] 
Consider a game on a network having a guardian $i$ with a dependant $j$. 
First consider an arbitrary equilibrium $x$ of this game. 
Clearly, we have three possibilities: $x_j=0$ ($j$ is a free rider), $x_j=e^*$ ($j$ is a specialist) or $0 < x_j < e^*$. Suppose $j$ is neither a specialist nor a free rider. By \eqref{eq:equilibrium efforts}, $\scn{j}{x} = x_j + x_i = e^*$, \ie, $i$ and $j$ are co-specialists. Hence a dependant  is either a specialist, free rider or co-specialist.
We now show that if $i$ has in addition to $j$, another dependant, say $k$, then $j$ cannot be a co-specialist. Observe that $k$ being a neighbour of co-specialist agent $i$ is a free rider by Lemma \ref{lem:basic results of LCP(G)}\ref{lem:co-specialists matching}. Hence $\scn{k}{x} = x_i < e^*$, which contradicts \eqref{eq:equilibrium efforts}. Hence $j$ can be either a specialist or a free rider but not a co-specialist. Since $d_j = 1$, from \eqref{lem:inequalities in efforts}, $\scn{j}{x} = e^*$. If $j$ is a free rider, then $\scn{j}{x} = x_i = e^*$, \ie, $i$ is specialist, whereby all dependants of $i$ are free riders, by Lemma \ref{lem:basic results of LCP(G)}\ref{lem:specialist IS}. And if $j$ is a specialist, $i$ is a free rider (again by  Lemma \ref{lem:basic results of LCP(G)}\ref{lem:specialist IS}), whereby from \eqref{eq:equilibrium efforts} every dependant $k$ of $i$ must exert effort $e^*$.  
Hence, dependants of $i$ are either all specialists (in which case $i$ is a free rider) or they are all free riders (in which case $i$ is a specialist).

Now consider a stable equilibrium. Bramoulle and Kranton showed that an equilibrium is stable only if it is specialized and all its free riders have at least two specialist neighbours (See \cite[Thm. 2]{bramoulle2007public}). Since the equilibrium is stable, it is necessarily specialized whereby the dependant $j$ is either a specialist or a free rider. But if $j$ is a free rider, it has only one specialist neighbour $i$ whereby the equilibrium is unstable. Hence every dependant of the network is necessarily a specialist for an equilibrium to be stable under best response dynamics.

For unstable equilibria, the results follow from the discussion of arbitrary equilibria.
\end{proof}

Observe that if a co-dependant pair exists, then only one of the two agents (who are both dependants) can be specialists whereby the equilibrium is always unstable.
For better clarity, the results in Table \ref{tab:dependants} are reorganized in Table \ref{tab:dependants2} putting in perspective the effect of the effort of a dependant at equilibrium on the stability of the equilibrium. Using the results in Tables \ref{tab:dependants} and \ref{tab:dependants2}, we have the result below regarding the non-existence of distributed equilibria in a network with at least one dependant.

\begin{table}[]
\centering

\begin{tabular}{|p{0.15\textwidth}|p{0.13\textwidth}|p{0.12\textwidth}|p{0.21\textwidth}|p{0.27\textwidth}|}
\hline & \multicolumn{3}{c|}{\large \it If a dependant is a \ldots} &  {\it \large and if it is \ldots}\\ \cline{2-4} & \multicolumn{1}{c|}{1. Specialist} & \multicolumn{1}{c|}{2. Free rider} & \multicolumn{1}{c|}{3. Co-specialist} & \multicolumn{1}{c|}{} \\ \hline
\multicolumn{1}{|c|}{\multirow{3}{0.15\textwidth}{\it \large Then the  equilibrium \ldots }} & may be stable & \multicolumn{1}{c|}{is unstable} & is unstable with suboptimal welfare $^{\#}$ &  1. the only dependant of its guardian        \\ \cline{2-5} 
\multicolumn{1}{|c|}{} & may be stable  & \multicolumn{1}{c|}{is unstable} & is not possible & 2. \textit{not} the only dependant of it guardian \\ \hline
\end{tabular}
\caption{\label{tab:dependants2} This table describes the effect of the efforts of dependant agents on the stability of the equilibrium.
\newline 
$\ ^{\#}$ A dependant can exist as a co-specialist in equilibrium only if it is the sole dependant of its guardian. However this is not only an unstable equilibrium, but also yields suboptimal welfare. We show this in Theorem \ref{thm:dependant_efficiency} later.}

\end{table}

\begin{theorem}\label{lem:dependant}{\bf (Free riders in networks with dependants)}
\\In a network with at least one dependant who is not a co-dependant, every Nash equilibrium always has a free-riding agent.
\end{theorem}
\begin{proof}
Suppose a guardian agent $i$ has multiple dependants. Then by Table \ref{tab:dependants} (Row 2), either all the dependants of $i$ are free riders or they are all specialists; in the latter case $i$ itself is a free rider. Hence there exists at least one free rider.

Now suppose the guardian $i$ has only one dependant $j$ such that $i$ and $j$ are not co-dependants, \ie, $i$ has another neighbour $k$. In this case, by Table \ref{tab:dependants} (Row 1), either (i) the $j$ is a free rider (and $i$ is a specialist) or (ii) $j$ is a specialist (and $i$ is a free rider), or (iii) $i$ and $j$ form a co-specialist pair. In the latter case, $k$ is a free rider according to Lemma \ref{lem:basic results of LCP(G)}\ref{lem:co-specialists matching}. Thus in every case, there always exists at least one free riding agent.
\end{proof}

\subsection{Forest networks}
We now consider the structure of equilibria on forests. A network of agents is called a forest if there is no cycle between them. We refer the reader to \cite{bondy1976graph} for a general introduction to forests and their properties. Here we recall that an induced subgraph of a forest is a forest and that a forest network can be represented as a disjoint union of trees (thus, for any two distinct trees in this union, there is no link in the forest having one vertex in each tree). 

A vertex in a graph is isolated if it has degree zero (\ie, no neighbours). A link is said to be isolated if both vertices in the link have degree unity. Note also that a tree which is not an isolated vertex has at least two dependants \cite{bondy1976graph}. 
 The following lemma is the interpretation of our previous graph theoretic results Lemma 6(c) and Lemma 9 from \cite{pandit2016linear}, in that order, in the context of public goods provision over networks. 

\begin{figure}
\includegraphics[width=\textwidth]{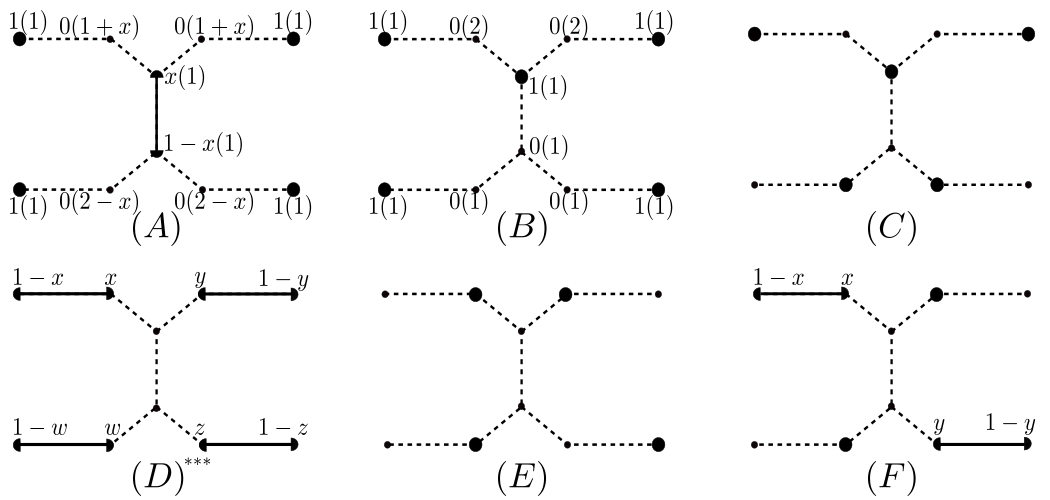}
\caption{\label{fig:trees}This figure shows six different equilibria of a public goods game over a tree network of 10 agents. A dotted line between agents indicates adjacency. Agents denoted by large black circles are specialists, whereas those denoted by small black circles are free riders. A solid line between adjacent agents indicates a co-specialist link, with black semicircles denoting the co-specialist agents. $(B),\ (C)$ and $(E)$ are specialized equilibria, whereas $(A),\ (D)$ and $(F)$ are hybrid equilibria.
For an agent $i$ of the network, the number denotes the normalized effort $\frac{1}{e^*}x_i$ exerted by it at equilibrium while the normalized effort of the closed neighbourhood $ \scn{i}{\frac{1}{e^*}x}$ is mentioned for equilibria $(A)$ and $(B)$ are mentioned in parenthesis. Observe that $\frac{1}{e^*}x_i \geq 0$, $\scn{i}{\frac{1}{e^*}x}\geq 1$ and at least one of the two inequalities is tight, as indicated by \eqref{eq:equilibrium efforts}.
\newline
$^{***}$ In equilibrium $(D)$, $x + y \geq 1$ and $w + z \geq 1$ are necessary for it to be an equilibrium.}
\end{figure}

\begin{lemma}\label{lem:trees}{\bf (Structure of equilibria in forest networks)}\\
For a public goods game over a forest network $G$,
\begin{enumerate}[label= (\alph*)]
\item \label{lem:trees0} If the game admits a distributed equilibrium, then the network is a disjoint union of isolated links and isolated agents.
\item \label{lem:trees1} The supporting agents in any equilibrium are either specialists or co-specialists.
\item \label{lem:trees2} If an equilibrium profile $x$ is not specialized, there exists another specialized equilibrium $y$ such that $\supp(y) \subset \supp(x)$. This specialized equilibrium incurs the same cost as the original equilibrium, \ie, $c\bfe \t y = c\bfe \t x$.
\item \label{lem:trees3} The cost of the equilibrium effort profile is $$ce^*\left({\rm \#(specialists) + \half\#(co-specialists)}\right),$$ where ${\rm \#(\cdot)}$ is stands for ``number of''.

\end{enumerate}
\end{lemma}

\begin{proof}
\begin{enumerate}[label = (\alph*)]
\item Consider a forest network $G$ which is a disjoint union of $m$ trees $T_1,\ldots,T_m$, such that $G$ admits a distributed equilibrium $x$. Then, since the trees $T_1,\hdots,T_m$ are disjoint, the subvector $x_{T_k}$ corresponding to efforts of agents in $T_k$ is a distributed equilibrium over $T_k$ for all $k$. 

For some $k$, let agent $j$ be a dependant in $T_k$ and let $i$ be the guardian of $j$. If $T_k$ has more than two agents, we must have $d_i > 1$ (since $T_k$ is a connected graph), whereby $i$ and $j$ are not co-dependants. Hence $T_k$ satisfies the condition in Lemma \ref{lem:dependant}. This contradicts the possibility of a distributed equilibrium $x_{T_k}$. Hence for any $k$, $T_k$ must consist of at most two agents. Hence $G$ consists of isolated links and isolated agents.

\item Let $x \in \NE(G)$. The network of supporting agents $G_{\supp(x)}$ is also a forest since it is an induced subgraph of the forest $G$. Let $x_{\supp(x)}$ denote the efforts exerted by the supporting agents. Observe that from Lemma \ref{lem:basic results of LCP(G)}\ref{lem:undisturbed}, it follows that $x_{\supp(x)}$ is an equilibrium for the public goods game over the forest $G_{\supp(x)}$. Moreover, this equilibrium is distributed. Now, from Lemma \ref{lem:trees}\ref{lem:trees0}, it follows that $G_{\supp(x)}$ consists only of isolated links and isolated vertices. The isolated links correspond to co-specialist links whereas the isolated vertices correspond to specialists. This proves the claim.

\item By Theorem \ref{lem:NE and LCP}, we know that $\frac{1}{e^*}x$ solves $\LCP(G)$. 
We know from Lemma \ref{lem:old lemmas}\ref{lem:old lemmas forest2} that if the network is a forest, there always exists a maximal independent set $S \subseteq \supp(\frac{1}{e^*}x)=\supp(x)$. The agents in $S$ form a specialized equilibrium given by $y:=e^*\bfone_S$, whereby $\supp(y) = S$. This proves the first part of the lemma. We now show that both equilibria $x$ and $y$ incur the same cost.

From Lemma \ref{lem:trees}\ref{lem:trees1} above, we know that the supporting agents in $x$ are either specialists or co-specialists. Let agent $k$ be a specialist in the equilibrium $x$. We first show that $k$ is also a specialist in $y$. Observe that if $k$ is not a specialist in $y$, then $k$ is a free rider. Hence $\scn{k}{y} = \sum_{k'\in \supp(y)}a_{kk'}y_{k'} \leq \sum_{k'\in \supp(x)}a_{kk'}y_{k'} = 0,$ since by Lemma \ref{lem:basic results of LCP(G)}\ref{lem:specialist IS}, $k$ has no neighbours in $\supp(x)$. This means that $y$ does not satisfy \eqref{eq:equilibrium efforts}; a contradiction to $y$ being an equilibrium. Every specialist in $x$  is a specialist in $y$.

Now let $i$ and $j$ be co-specialists in $x$. We now show that exactly one of $i$ and $j$ specializes in $y$ whereas the other free rides. By Lemma \ref{lem:basic results of LCP(G)}\ref{lem:specialist IS}, it is clear that both $i$ and $j$ cannot be specialists in $y$, since they are adjacent in $G$. On the other hand, suppose both $i$ and $j$ are free riders in $y$. Then, $\scn{i}{y} = \sum_{k'\in \supp(y)}a_{ik'}y_{k'} \leq \sum_{k'\in \supp(x)}a_{ik'}y_{k'} = y_j + \sum_{k'\in \supp(x)\backslash \{j\}}a_{ik'}y_{k'} = 0,$ since by Lemma \ref{lem:basic results of LCP(G)}\ref{lem:co-specialists matching} neighbours of co-specialists free ride, whereby $i$ has no neighbours in $\supp(x)$ other than $j$.

Thus, out of every pair of co-specialists in $x$, one agent specializes in the specialized equilibrium $y$, while the other free rides. This substitution of effort between co-specialists requires the same total effort. Hence, the equilibrium cost 
\begin{equation}
c\bfe \t x = c \left (\#({\rm specialists\ in\ } x)  + \half \#({\rm co-specialists\ in\ } x)\right )  =c \# \left ( {\rm specialists\ in\ } y \right )= c \bfe \t y, \label{eq:forest cost} 
\end{equation}
\ie, the cost incurred by both equilibria is the same.

\item This was shown in \eqref{eq:forest cost} above.
\end{enumerate}
\end{proof}

Lemma \ref{lem:trees}\ref{lem:trees1} means that in an equilibrium $x$ in a public goods game over a forest, if an agent $i$ is neither a specialist nor a free rider, \ie, $0 < x_i < e^*$, then there always exists a neighbour $j$ of $i$ that is a co-specialist of $i$ in equilibrium $x$ and, in another equilibrium, substitutes the deficit of the effort of $i$. In the specialized equilibrium described in part \ref{lem:trees2}, in every co-specialist pair, one agent specializes to substitutes the effort of its co-specialist, who then free rides. The results in Lemma \ref{lem:trees} are summarized in Table \ref{tab:trees}. 

In Figure \ref{fig:trees}, various equilibria on a tree network are shown. Notice that $(B)$ and $(E)$ are specialized equilibria contained in the support of equilibria $(A)$ and $(D)$ respectively. 

\begin{table}[]
\centering

\begin{tabular}{|p{0.14\textwidth}|p{0.1\textwidth}|p{0.30\textwidth}|p{0.15\textwidth}|}
\hline
\multirow{2}{*}{}                                               & \multicolumn{3}{c|}{\it \large If network is \ldots}                                                                                    \\ \cline{2-4} 
\multicolumn{1}{|l|}{\it \large Then \ldots} & \multicolumn{1}{c|}{Any network}     & \multicolumn{1}{c|}{Forest} & \multicolumn{1}{c|}{Well-covered forest} \\ \hline
Supporting agents are & \multicolumn{1}{c|}{not free riders} & \multicolumn{1}{c|}{specialists or co-specialists} & \multicolumn{1}{c|}{specialists or co-specialists}\\ \hline
Total equilibrium cost is & \multicolumn{1}{c|}{$c \ \bfe \t x$} & \multicolumn{1}{c|}{$ce^*\left({\rm \#(specialists) + \half\#(co-specialists)}\right)$} & \multicolumn{1}{c|}{$\half ce^* n$ $^{\dag\dag}$}  \\ \hline
\end{tabular}
\caption{Summary of roperties of supporting agents as described in Lemma \ref{lem:trees}. \newline $^{\dag\dag}$ Follows from Theorem \ref{thm:well-covered}\ref{thm:well-covered cost}, proved later in Section 6.}\label{tab:trees} 

\end{table}

\section{Refinement of the equilibrium based on weighted total  effort and cost}

In this section, we prove Theorem \ref{thm:total effort} and compare equilibria based on their total weighted effort and the total cost incurred, and study how the structure of the network affects the existence of an optimum specialized equilibrium. The cost incurred by an equilibrium is $c\bfe \t x$, whereas for a set of weights $w_i \geq 0$, the $w$-weighted effort is $w \t x$. Recall from the discussion in the introduction that although the above functions for ranking equilibria are linear, finding the optimum equilibrium is a hard problem.

\subsection*{Proof of Theorem \ref{thm:total effort}}\begin{enumerate}[label=(\alph*)]
\item Let $S^*$ be the maximum $w$-weighted independent set in $G$. We need to show that $e^*\bfone_{S^*}$ is a specialized equilibrium that attains the maximum $w$-weighted equilibrium effort. Let $x$ be an equilibrium profile. By Theorem \ref{lem:NE and LCP} we know that $\frac{1}{e^*}x \solves \LCP(G)$, whereby from Theorem \ref{thm:old results}\ref{thm:wTx} we can say that $$\EWstar = \max\left\{ w \t x \newmid x \in {\rm NE}(G)\right\} = \max\left\{ w \t e^* y \newmid y \solves \LCP(G)\right\} = e^*\alpha_w(G).$$  Since $w_i \geq 0$ for all $i\in V,$ $S^*$ is also a maximal independent set.  As a consequence $e^*\bfone_{S^*}$ is a specialized equilibrium, and its weighted effort is $e^* w \t \bfone_{S^*} = e^*\alpha_w(G) $, by definition of the $w$-weighted maximum independent set. It follows that $e^*\bfone_{S^*} $ attains the maximum total weighted equilibrium effort and hence $\EWstar = \EWSstar$.

\item Observe that the minimum cost in equilibrium is given by, $$ \Cstar = c\cdot \min\left\{\bfe \t x \newmid x \in {\rm NE}(G) \right\} = c\cdot e^* \min\left\{\bfe \t x \newmid x\ {\rm solves} \ \LCP(G)\right\}.$$
 By Theorem \ref{thm:old results}\ref{thm:beta equality} we can say that if the network $G$ is a forest, the above quantity is $c e^*\beta(G)$. Let $S'$ denote the smallest maximal independent set of $G$, then $e^*\bfone_{S'}$ is a specialized equilibrium which incurs a cost $ce^* \bfe \t \bfone_S = ce^*\beta(G)$. Hence if the network is a forest, we have that $\Cstar = \CSstar$, and $e^*\bfone_{S'}$ incurs the minimum cost.

\item We first show a more general result: If $y$ is an equilibrium and a subset $S\subseteq \supp(y)$ of the supporting agents of $y$ form a maximal independent set of the network, then the specialized equilibrium supported on $S$ requires total effort at least as much as that of $y$, \ie, $e^* |S| \geq \sum_j y_j$. Let $U := \supp(y) \backslash S$.  Then, from \eqref{eq:equilibrium efforts}, $\forall\; i \in \supp(y),$
$$ \quad \scn{i}{y} = \sum_{j \in V}a_{ij}y_j + y_i = e^*.$$
Summing over $i\in S$ gives,
\[ \sum_{i \in S}\sum_{j\in V}a_{ij}y_j +  \bfe \t y - \sum_{j \in U}y_j = e^*|S|.\]
Thus,
\begin{equation*}
e^*|S| - \bfe \t y
\stackrel{(a)}{=} \sum_{i\in S}\sum_{j\in U} a_{ij} y_j  - \sum_{j \in U}y_j
= \sum_{j \in U} (|N_S(j)| - 1)y_j \stackrel{(b)}{\geq} 0  .
\end{equation*}
The equality in $(a)$ follows from the observation that $a_{ij}$ is nonzero for $i \in S$ only if $j \notin S$ due to $S$ being an independent set. Moreover, in this set $V \backslash S$, $y_j$ is non-zero only for $j \in U$. Now, since $S$ is a maximal independent set, every vertex not in $S$ has at least one neighbour in $S$ whereby $|N_S(j)| \geq 1$ and justifies the inequality $(b)$. 

Now, if the network is such that it admits a distributed equilibrium, all the agents in the network are supporting agents of this equilibrium whereby the support of any specialized equilibrium is clearly its subset. Following the discussion above, we can say that the cost of a distributed equilibrium is at most as much as that of any specialized equilibrium, \ie, $\CDstar \leq \CSstar$.


\item By Lemma \ref{lem:old lemmas}\ref{lem:regular}, we know that if the network is regular, $\frac{\bfe}{d+1}$ is a solution to $\LCP(G)$ such that it minimizes the function $\bfe \t x$ amongst all solutions of $\LCP(G)$. It follows from Theorem \ref{lem:NE and LCP} that $e^*\frac{\bfe}{d+1}$ is a distributed equilibrium of the public goods game of the regular network with minimum cost, whereby $\Cstar = \CDstar$. This proves the Theorem \ref{thm:total effort}\ref{thm:total cost_regular}.
\end{enumerate}

\section{Refinement of the equilibrium based on welfare}\label{sec:refine}

In this section we prove Theorem \ref{thm:main} by studying the welfare of specialized equilibria. We establish conditions on the concavity of the benefit function and on the structure of the underlying network for which there exists a specialized equilibrium which yields maximum equilibrium welfare. Moreover, we also give the specialized equilibrium which yields optimal welfare, under these conditions.

We show that under certain conditions on the concavity of the benefit function, a limiting result is possible for all networks. However if the class of networks is narrowed to forests, then the result holds for a larger class of benefit functions $b$. Hence specialized equilibria may be considered a refinement of the Nash equilibrium  of the public goods game over a network, when searching for welfare maximizing equilibria.

While proving Theorem \ref{thm:main} we first show bounds on the welfare of equilibrium profiles thereby establishing a range for the maximum equilibrium welfare. We then show the convergence of these bounding functions to identify the behaviour of the welfare function and the maximum equilibrium welfare as the concavity of the benefit function varies. While varying the concavity of the benefit function, we keep $b(e^*)$ and $b'(e^*)$ fixed whereby the equilibra remain unchanged.
  
To this end, for $\mu \in \Real^n$, define, \[\theta_{\mu} := \max\{ \mu \t x \mid x \in \NE(G) \} \qquad \aur\qquad  \theta_\mu^{\mathsf{S}} := \max\{ \mu \t x \mid x \in \SNE(G) \}.\]
\begin{lemma} \label{lem:continuity of theta} For $\mu \in \Real^n$, $\theta_{\mu}$ and $\theta_\mu^{\mathsf{S}}$ are continuous functions of $\mu$.
\end{lemma}
\begin{proof}
Observe that $\theta_{\mu}$ and $\theta_{\mu}^{\mathsf{S}}$ are both value functions of the optimization of a continuous function $\mu\t x$ over sets $\NE(G)$ and $\SNE(G)$ that are compact as well as independent of $\mu$. If follows from stability results in optimization theory, such as \cite[Thm. 7]{hogan73point}, that both $\theta_{\mu}$ and $\theta_{\mu}^{\mathsf{S}}$ are continuous.
\end{proof}

For the rest of the section, we assume that the network does not have any isolated agents,  \ie, $d_i \geq 1$ for all $i$ as assumed in the statement of the theorem. Define 
{\color{black} $$\sigma_j:= \begin{cases} \frac{b(d_j e^*) - b(e^*)}{ce^*(d_j-1)} &\eef d_j > 1, \\
\sigma_b & \eef d_j = 1.
\end{cases} \qquad \aur \qquad \sigma_j':= \begin{cases} \frac{1}{c}b'(d_j e^*) &\eef d_j > 1, \\
\sigma_b & \eef d_j = 1.
\end{cases}$$  If $d_j> 1$, observe that $\sigma_j$ denotes the normalized\footnote{Normalization refers to division by $c$.} slope of the secant between $(e^*,b(e^*))$ and $(d_j e^*,b(d_j e^*))$, while $\sigma'_j$ denotes the normalized slope of the tangent to the benefit function $b$ at the point $d_j e^*$. These have been depicted in Figure \ref{fig:bounds using concavity} for clarity.
Let $l_j = (\sigma_j   + \sum_{i }a_{ij}\sigma_i  - 1)$ and  $u_j = (\sigma_j '   + \sum_{i }a_{ij}\sigma_i'  - 1)$. Denote by $\sigma,\ \sigma ',\ l$ and $u$ the corresponding vectors with $n$ components.

In establishing Theorem~\ref{thm:main}, we need to be able to formally relate $\WU^*$ to changes in $\sigma_b.$ We also note that while the limiting behaviour of $\WU$ as $\sigma_b \rightarrow 1$ is known from~\cite[Prop 1]{bramoulle2007public}, it does not automatically yield a proof of Theorem~\ref{thm:main}. The following theorem provides  linear upper and lower bounds on the welfare of equilibrium effort profiles and the maximum equilibrium welfare. These bounds are essential in the proof of Theorem~\ref{thm:main}. }

\begin{theorem}\label{lem:linear bounds}
For a public goods game over a network without isolated agents, if $x$ is an equilibrium effort profile,
\begin{enumerate}[label=(\alph*)]
\item \label{lem:linear bounds1}
The welfare function is bounded as follows,
\begin{equation}\label{eq:bounds}
c l\t x   \leq \WU(x) - nb(e^*) + ce^*\bfe \t \sigma \leq  c\ u \t x + ce^*d \t \big(\sigma - \sigma'\big),
\end{equation}
whereby,
\begin{equation}\label{eq:bound on WU*}
c\theta_{l}  \leq \WU^* - nb(e^*) + ce^* \bfe \t \sigma  \leq c\theta_{u} + c e^*d \t \big(\sigma - \sigma' \big),
\end{equation}
and similarly ,
\begin{equation}\label{eq:bound on WUS*}
ce^*\theta_{l}^{\mathsf{S}}  \leq \WUSstar - nb(e^*) + ce^* \bfe \t \sigma  \leq ce^*\theta_{u}^{\mathsf{S}} + c e^*d \t \big(\sigma - \sigma' \big).
\end{equation}

\item \label{thm:limits}Keeping $b'(e^*)$ and $b(e^*)$ fixed and varying $\sigma_b$, we have that,
\begin{equation}
\lim_{\sigma_b \rightarrow 1} l_j = \lim_{\sigma_b \rightarrow 1} u_j = d_j,\qquad {\rm\ for\ all\ }j,
\end{equation}
\begin{equation}
\lim_{\sigma_b \rightarrow 0} l_j = \lim_{\sigma_b \rightarrow 0} u_j = -1,\qquad {\rm\ for\ all\ }j,
\end{equation}
\begin{equation}
\lim_{\sigma_b \rightarrow 0} \sigma_j-\sigma_j' = \lim_{\sigma_b \rightarrow 1} \sigma_j-\sigma_j' = 0,\qquad {\rm\ for\ all\ }j.
\end{equation}
\item Keeping $b'(e^*)$ and $b(e^*)$  fixed and varying $\sigma_b$, we have that, 
\label{lem:linear bounds2}
\begin{align*}
\lim_{\sigma_b \rightarrow 1}\WU(x) &= n\big(b(e^*) - ce^*\big)  + c d \t x, \\
\lim_{\sigma_b \rightarrow 0}\WU(x) &= nb(e^*) - c \bfe\t x. 
\end{align*}
\end{enumerate}
\end{theorem}

\begin{proof}
\begin{enumerate}[label = (\alph*)]
\item From Lemma from \ref{lem:basic results of LCP(G)}\ref{lem:bounds on C_j}, we know that for an equilibrium profile of efforts $x$, $e^* \leq \scn{i}{x} \leq d_j e^*$ for agents $j$ with $d_j>1$. Due to the concavity of the benefit function, the tangent at $d_j e^*$ always lies above the function whereas the secant between $(e^*,b(e^*))$ and $(d_j e^*,b(d_j e^*))$ always lies below the function for the interval $[e^*, d_j e^*]$ (see Figure \ref{fig:bounds using concavity}). Hence we have, 
\begin{align}\label{eq:welfare bound}
b(e^*) + c\sigma_j (\scn{j}{x} - e^*) \quad &\leq  \quad b(\scn{j}{x}) \quad  \leq  \quad  b(d_j e^*) - c\sigma_j' (d_j e^* - \scn{j}{x}) ,\\
\iff  \qquad\qquad c \sigma_j \scn{j}{x} \quad &\leq  \quad b(\scn{j}{x}) - b(e^*) + c\sigma_j e^*\quad  \leq  \quad  c\sigma_j ' \scn{j}{x} + ce^*d_j(\sigma_j - \sigma_j ')\label{eq:welfare bound2},
\end{align}
where the equivalence follows from the equation $b(d_j e^*) = b(e^*) + c \sigma_j e^* (d_j - 1)$ and subtracting $b(e^*)-c\sigma_je^*$ from all sides. For the case where $d_j = 1$, it can be seen from Lemma \ref{lem:basic results of LCP(G)}\ref{lem:bounds on C_j}, that $\scn{j}{x} = e^*$ whereby $b(\scn{j}{x}) = b(e^*)$, \ie, \eqref{eq:welfare bound} holds for agents who are dependants.
Now since $\WU(x) = \sum_{j \in V}\big(b(\scn{j}{x}) - cx_j\big)$, the sum of the inequalities in \eqref{eq:welfare bound2} for all agents $j$, gives \begin{align*}
c\sum_j (\sigma_j\scn{j}{x} - x_j) \leq \WU(x) -nb(e^*) + c e^* \bfe \t \sigma\leq c\sum_j(\sigma_j'\scn{j}{x} - x_j) + ce^*d_j(\sigma_j - \sigma_j').
\end{align*}
Observe that, \begin{align*}\sum_j\big( \sigma_j\scn{j}{x} - x_j \big)&:= \sum_j\big( \sigma_j(x_j + \sum_i a_{ij}x_i) - x_j\big) = \sum_j(\sigma_j -1)x_j + \sum_j \sigma_j \sum_i a_{ij} x_i \\
&\stackrel{(c)}{=} \sum_j(\sigma_j -1)x_j + \sum_i x_i \sum_j a_{ij}\sigma_j 
\stackrel{(d)}{=} \sum_j \big(\sigma_j -1 + \sum_i a_{ij}\sigma_i\big)x_j = l \t x,
\end{align*}
where the equality in $(c)$ is due to interchanging the order of summation and the equality in $(d)$ holds by exchange of summation indices $i$ and $j$. Following a similar argument,we get $\sum_j \sigma_j'\scn{j}{x} - x_j = u \t x$. The bounds on $\WU(x) - nb(e^*) + ce^*\bfe\t\sigma$ follow directly. Moreover, the bounds on the maximum equilibrium welfare $\WU^*$ in \eqref{eq:bound on WU*} and the maximum specialized equilibrium welfare $\WUSstar$ in \eqref{eq:bound on WUS*} follow after maximizing all quantities in \eqref{eq:bounds} over $\NE(G)$ and $\SNE(G)$ , respectively.

\item We first show the intermediate limits, $$\lim_{\sigma_b \rightarrow 1} \sigma_j = \lim_{\sigma_b \rightarrow 1} \sigma_j' = 1, \aur \lim_{\sigma_b \rightarrow 0} \sigma_j = \lim_{\sigma_b \rightarrow 0} \sigma_j' = 0,\qquad {\rm\ for\ all\ }j.$$ Since $l_j$, $u_j$ and $\sigma_j - \sigma_j'$ are linear functions of $\sigma$ and $\sigma'$, by the sum law of limits, showing the above limits is sufficient to prove the limits in the statement of the theorem.

\textbf{Case I}: ($d_j=1$) Clearly, by definition $\sigma_j=\sigma_j' = \sigma_b$, whereby the limits hold trivially for both $\sigma_b\rightarrow 0$ and $\sigma_b\rightarrow 1$.

\textbf{Case II}: ($d_j>1$) First consider the case $\sigma_b \rightarrow 1.$
To show the limit of $\sigma_j$, observe that $\sigma_j \geq \sigma_b$ since $b$ is increasing and concave. Moreover, by the concavity of $b$, the tangent at $e^*$ lies above the function, \ie,
\begin{equation}
b'(e^*) (d_j-1)e^* +b(e^*)\geq b(d_je^*). \label{eq:tgtineq} 
\end{equation}
Since $b'(e^*) = c$, rearranging, we have $\sigma_j \leq 1$. Now since $\sigma_j \geq \sigma_b,$ we get $\lim_{\sigma_b \rightarrow 1}\sigma_j = 1$. 

For computing $\lim_{\sigma_b\rightarrow 1}\sigma_j'$, let $\sigma^*$ denote the normalized slope of the secant between $(e^*,b(e^*))$ and $\big((n-1)e^*,b((n-1)e^*)\big)$. Observe that since $b$ is increasing and concave, we again have $1\geq \sigma^*\geq \sigma_b$, whereby we have $\lim_{\sigma_b \rightarrow 1}\sigma^*=1$.

Since $1 \leq d_j \leq n-1$ for any $j$, by the concavity of $b$, we have,
\begin{equation}
b'(e^*)\geq b'(d_je^*) \geq b'\big((n-1)e^*\big). \label{eq:decreasingder} 
\end{equation}
Moreover, an inequality similar to \eqref{eq:tgtineq} with the tangent to $b$ at $(n-1)e^*$ gives,
\begin{align*}
b'((n-1)e^*) \geq \frac{1}{e^*}(b(ne^*) - b((n-1)e^*))
= \frac{1}{e^*}(b(e^*) + e^*c(n-1)\sigma_b - (b(e^*) + e^*c(n-2)\sigma^*)),
\end{align*}
whereby $\sigma_j' \leq 1.$
 Hence combining with \eqref{eq:decreasingder} gives, $\lim_{\sigma_b \rightarrow 1}b'(d_je^*)\geq  c\lim_{\sigma_b \rightarrow 1} \big((n-1)\sigma_b-(n-2)\sigma^*)\big)=c$, whereby, $\lim_{\sigma_b\rightarrow 1}\sigma_j'  \geq 1$. Since $\sigma_j' \leq 1$,  $\lim_{\sigma_b\rightarrow 1}\sigma_j' = 1$.

To show the limits of $\sigma_j$ and $\sigma_j'$ as $\sigma_b\rightarrow 0$, for $d_j > 1$, we need the following set of inequalities, 
\begin{equation}\label{eq:Thm1_temp2}
0 \leq \sigma_j ' \leq\sigma_j \leq \sigma_b \left( \frac{n-1}{d_j -1}\right).
\end{equation}
Observe that the first and second inequality follow from the monotocity and concavity of $b$. The third inequality follows from $0 \leq b(n e^*) - b(d_j e^*)= b(n e^*) -b(e^*) - \big(b(d_j e^*)-b(e^*)\big)= \sigma_b(n-1)e^* - \sigma_j(d_j - 1)e^*$. Hence $\lim_{\sigma_b \rightarrow 0} \sigma_j = \lim_{\sigma_b \rightarrow 0} \sigma_j' = 0$. 

\item 
 Applying limits to the result in part \ref{lem:linear bounds1}, we have that $$\lim_{\sigma_b\rightarrow 1}c l \t x \leq \lim_{\sigma_b\rightarrow 1}\WU(x) - nb(e^*) + ce^*\lim_{\sigma_b \rightarrow 1}\sum_{j}\sigma_j \leq \lim_{\sigma_b\rightarrow 1} c u \t x + \lim_{\sigma_b\rightarrow 1}ce^*d\t(\sigma-\sigma'),$$ and $$\lim_{\sigma_b\rightarrow 0}c l \t x \leq \lim_{\sigma_b\rightarrow 0}\WU(x) - nb(e^*) + ce^*\lim_{\sigma_b \rightarrow 0}\sum_{j}\sigma_j \leq \lim_{\sigma_b\rightarrow 0} c u \t x + \lim_{\sigma_b\rightarrow 0}ce^*d\t(\sigma-\sigma').$$ Applying the limits from part \ref{thm:limits} proves the result.

 \end{enumerate}
\end{proof}

\begin{figure}
\center
\includegraphics[width=0.6\textwidth]{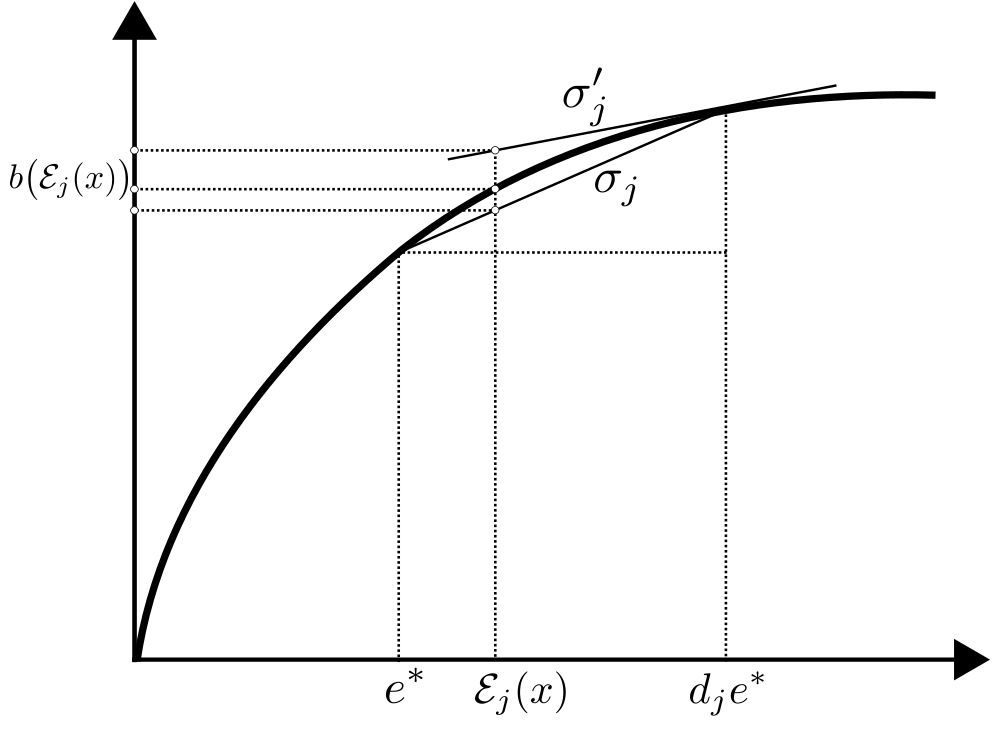}
\caption{\label{fig:bounds using concavity} Bounds on $b(\scn{i}{x})$ for $j$ such that $d_j > 1$ as given by the inequality in \eqref{eq:welfare bound}. The tangent and secant have slopes $c \sigma_j '$ and $c \sigma_j$ respectively. Due to concavity of the function $b$, the tangent always lies above the function and the secant always lies below the function for the interval $[e^*,d_je^*]$. Hence $b(e^*) +c \sigma_j \big(\scn{j}{x} - e^*\big) $ and $b(d_je^*) - c\sigma'_j \big(d_j e^* - \scn{j}{x}\big)$ are the bounds. In the case where $d_j = 1$, both quantities become $b(e^*)$}
\end{figure}

\subsection*{Proof of Theorem \ref{thm:main}}

\begin{enumerate}[label=(\alph*)]
\item Let $x^*$ be the specialized equilibrium as mentioned in the statement of the theorem, whereby $x^* = e^*\bfone_S$, using Theorem \ref{lem:NE and LCP}.
The limiting value of the welfare at this equilibrium, using Theorem \ref{lem:linear bounds}\ref{lem:linear bounds2} is, 
\begin{equation}\label{eq:thm1(a)temp1}
\lim_{\sigma_b \rightarrow 1}\WU(x^*) = n(b(e^*) - ce^*) + ce^*\sum_{j\in S}d_j = n(b(e^*) - ce^*) + ce^*\alpha_d(G),
\end{equation}
since $S$ is a $d$-weighted maximum independent set in the network. We now calculate the limiting value of $\WU^*$ using \eqref{eq:bound on WU*}, \ie, $$ \lim_{\sigma_b \rightarrow 1}c\theta_l \leq \lim_{\sigma_b \rightarrow 1} \WU^* - nb(e^*) + ce^* \bfe \t \sigma \leq \lim_{\sigma_b \rightarrow 1} c\theta_u .$$ However, as a consequence of Theorem \ref{lem:linear bounds}\ref{thm:limits}, $\lim_{\sigma_b \rightarrow 1}l = \lim_{\sigma_b \rightarrow 1}u = d$, whereby using Lemma \ref{lem:continuity of theta}, $\lim_{\sigma_b \rightarrow 1}\theta_l=\lim_{\sigma_b \rightarrow 1}\theta_u=\theta_d$. Moreover,
since $d\geq 0$, $\theta_d = e^*\alpha_d(G)$ by Theorem \ref{thm:old results}\ref{thm:wTx}. Hence 
\begin{equation}
\label{eq:thm1(a)temp2}
\lim_{\sigma_b \rightarrow 1}\WU^* = 
n\big(b(e^*) - ce^*\big) + ce^*\alpha_d(G).
\end{equation}
Hence $\lim_{\sigma_b \rightarrow 1} \WU^*=\lim_{\sigma_b \rightarrow 1}\WU(x^*).$
Now, since $\WU^* \geq \WUSstar \geq \WU(x^*)$, taking limits on all three quantities gives
$$\lim_{\sigma_b \rightarrow 1}\WU^* \geq \lim_{\sigma_b \rightarrow 1}\WUSstar \geq \lim_{\sigma_b \rightarrow 1}\WU(x^*)$$
Since the limiting values of $\WU^*$ and $\WU(x^*)$ are the same, the above inequality proves the result.

\item Consider a forest $G$ and let $x'$ be the specialized equilibrium as mentioned in the statement of the theorem, whereby $x' = e^*\bfone_{S'}$, using Theorem \ref{lem:NE and LCP}.
The limiting value of the welfare at this equilibrium, using Theorem \ref{lem:linear bounds}\ref{lem:linear bounds2} is, 
\begin{equation}\label{eq:thm1(b)temp1}
\lim_{\sigma_b \rightarrow 0}\WU(x') = nb(e^*) - ce^*\sum_{j\in S'}1 = nb(e^*) - ce^*\beta(G),
\end{equation}
since $S'$ is the smallest maximal independent set in the network, \ie, cardinality of $S'$ is $\beta(G)$. Now, to calculate the limiting value of $\WU^*$, applying limits to \eqref{eq:bound on WU*} gives, $$ \lim_{\sigma_b \rightarrow 0}c\theta_l \leq \lim_{\sigma_b \rightarrow 0} \WU^* - nb(e^*) + ce^* \bfe \t \sigma \leq \lim_{\sigma_b \rightarrow 0} c\theta_u.$$ However, from Theorem \ref{lem:linear bounds}\ref{thm:limits} $\lim_{\sigma_b \rightarrow 0}l=\lim_{\sigma_b \rightarrow 0}u=-\bfe$, whereby using Lemma \ref{lem:continuity of theta}, we have $\lim_{\sigma_b \rightarrow 0}\theta_l=\lim_{\sigma_b \rightarrow 0}\theta_u=\theta_{-\bfe}$. 
Now observe that, $$-\theta_{-\bfe}=  \min \{ \bfe \t x \mid x \in \NE(G)\} = e^*\min \{ \bfe \t y \mid y \solves \LCP(G)\}.$$ 
Moreover, since $G$ is a forest, we know from Theorem \ref{thm:old results}\ref{thm:beta equality}, that $\min \{ \bfe \t x \mid x \solves \LCP(G)\} = \beta(G)$. Hence
\begin{equation}
\label{eq:thm1(b)temp2}
\lim_{\sigma_b \rightarrow 0}\WU^* = 
nb(e^*) - ce^*\beta(G),
\end{equation}
whereby $\lim_{\sigma_b \rightarrow 0} \WU^* = \lim_{\sigma_b \rightarrow 0} \WU(x').$
Since $\WU \geq \WUSstar \geq \WU(x')$, taking limits on all three quantities and arguing as in part (a) gives the result.

\item   Let the network $G$ admit a distributed equilibrium $x^*$, \ie, $x_i^*>0$ for all $i$. It follows from \eqref{eq:equilibrium efforts} that for any $x^*$, $\scn{i}{x} = e^*$ for all agents $i$. Hence $\WU(x^*) = nb(e^*) - c \bfe \t x^*$, which is independent of variation in $\sigma_b$ (recall that keeping $b(e^*)$ and $b'(e^*)$ fixed as $\sigma_b$ varies ensures that the equilibria are unaltered).

Let $S$ be a maximal independent set of $G$, whereby $e^*\bfone_S$ is a specialized equilibrium. From Theorem \ref{thm:total effort}\ref{thm:total cost}, the cost  $c\bfe\t x^*$ incurred by the distributed equilirium is at most as much as the cost $ce^*|S|$ incurred by the specialized equilibrium $e^*\bfone_S$, \ie, $ce^*|S| \geq c\bfe \t x^*$ or equivalently, $$\WU(x^*) \geq nb(e^*) -ce^*|S|.$$
 Observe that this holds true for any $x^* \in \DNE(G)$ and any maximal independent set $S$, \ie, the variable $S$ in the RHS is independent of the variable $x^*$ in the LHS. Hence,  we may infimize over $x^*$ and maximize over $S$, leading to,
\begin{align}
\label{eq:thm1(c)temp1}
\inf \{\WU(x^*)\mid x^* \in \DNE(G)\} &\geq \max \{nb(e^*)-ce^*|S| \mid e^*\bfone_S \in \SNE(G)\},\nonumber\\
\WUDstar &\geq  nb(e^*) - ce^*\min\{|S|\mid S \ {\rm is\ maximal\ independent}\},\nonumber\\
&=nb(e^*) - ce^*\beta(G).
\end{align}

The limiting value of $\WUSstar$ using \eqref{eq:bound on WUS*} is,
$$ \lim_{\sigma_b \rightarrow 0}c\theta_l^{\mathsf{S}} \leq \lim_{\sigma_b \rightarrow 0} \WUSstar - nb(e^*) + ce^* \bfe \t \sigma \leq \lim_{\sigma_b \rightarrow 0} c\theta_u^{\mathsf{S}}.$$ Now, as in part (b), $\lim_{\sigma_b \rightarrow 0}l=\lim_{\sigma_b \rightarrow 0}u=-\bfe$,  whereby $\lim_{\sigma_b \rightarrow 0} \WUSstar = nb(e^*) + c\theta_{-\bfe}^{\mathsf{S}} $. 
Moreover, again as in part (b), 
$$-\theta_{-\bfe}^{\mathsf{S}}=  \min \{ \bfe \t x \mid x \in \SNE(G)\} = e^*\min \{ |S| \mid S {\rm \ is\ maximal\ independent}\} = e^*\beta(G),$$ since $\beta(G)$ is the cardinality of the smallest maximal independent set of $G$.
 Hence,
\begin{equation}
\label{eq:thm1(c)temp2}
\lim_{\sigma_b \rightarrow 0}\WUSstar = nb(e^*) - ce^*\beta(G).
\end{equation}

The statement of the theorem follows from \eqref{eq:thm1(c)temp1} and \eqref{eq:thm1(c)temp2}.

\end{enumerate}

\section{Welfare and cost in well-covered forest networks}

In this section we prove Theorem \ref{thm:well-covered}. A network $G$ is said to be well-covered if all its maximal independent sets have the same cardinality, \ie, $\alpha(G) = \beta(G)$. Figure \ref{fig:well_covered} shows example networks which are well-covered.

Before proceeding to the proof, we show a more general result regarding the role of dependants on the efficiency of the equilibrium. This result is put to use while proving Theorem \ref{thm:well-covered}\ref{thm:well-covered welfare}.

\begin{figure}
\centering
\includegraphics[width=0.4\textwidth]{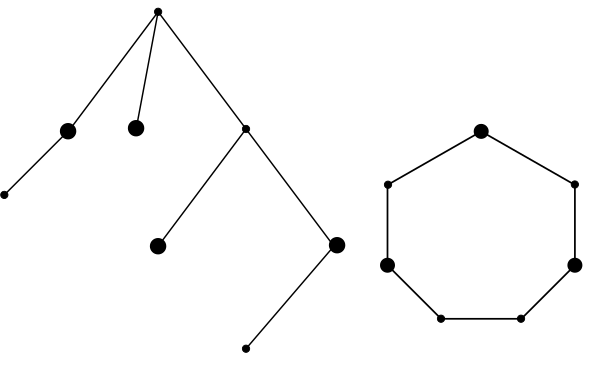}
\caption{\label{fig:well_covered}Examples of well-covered networks. The vertices marked by larger circles form maximal independent sets. All maximal independent sets in the well-covered tree network consist of 4 agents while those in the cycle network consist of 3 agents respectively. Observe that the well-covered tree over 8 agents depicted in the figure has exactly 4 guardian-dependant pairs such that exactly one of the two is part of any maximal independent set.}
\end{figure}

\begin{theorem}
\label{thm:dependant_efficiency}
{\bf (Co-specialist pendant lines are inefficient)}\\
In a public goods game, if a guardian has a single dependant, then in an equilibrium yielding maximum equilibrium welfare, exactly one of the guardian-dependant pair specializes whereby the other free rides.
\end{theorem}

\begin{proof} Recall from Table \ref{tab:dependants} (Row 1, Column 1), that if a guardian has a single dependant, then either one of the guardian or dependant free rides or they both form a co-specialist pair. Note that an equilibrium where the guardian-dependant pair form a co-specialist link is not possible if the guardian has multiple dependants (See Table \ref{tab:dependants2} (Row 2, Column 3)). We have to show that if the guardian-dependant pair form a co-specialist link, then the equilibrium yields suboptimal welfare. We affirm this by showing the existence of another equilibrium which yields higher welfare.

Let agent $i$ be the only dependant of agent $j$ such that they are not co-dependants. Let $x$ be an equilibrium such that $i$ and $j$ form a co-specialist link, \ie, $x_i + x_j = e^*$. From Lemma \ref{lem:basic results of LCP(G)}\ref{lem:co-specialists matching}, we can infer that all the neighbours of $j$ other than $i$ free ride in the equilibrium $x$, \ie, \begin{equation}\label{eq:temp6.3.0}x_k = 0, \qquad \forall\ k \in N_j\backslash\{i\}.\end{equation} Now, consider the profile of efforts $y$, where $y_j = e^*$, $y_i = 0$ and for all other agents $k$ exert the same effort $y_k = x_k$. We first show $y \in {\rm NE}(G)$ and then prove that indeed $W_U(y) > W_U(x)$.

Observe that for all agents $k$ other than $j$ and its neighbours, the effort of closed neighbourhood remains unchanged, \ie, $y_k = x_k, \aur \scn{k}{y}=\scn{k}{x}, $ whereby,\begin{equation}\label{eq:temp6.3.1}
 y_k\geq0,\ \scn{k}{y}\geq e^*,\ y_k(\scn{k}{y}- e^*)=0,\qquad \forall\ k \notin \{j\}\cup N_j,
\end{equation} from \eqref{eq:equilibrium efforts} since $x$ is an equilibrium. 
Now, for agent $j$, $$y_j=e^*,\aur \scn{j}{y} = y_j + y_i + \sum_{k \in N_j\backslash\{i\}}y_k = e^*+0+\sum_{k \in N_j\backslash\{i\}}x_k=e^*,$$ by \eqref{eq:temp6.3.0}. Moreover we have that, $$y_k = x_k = 0, \aur \scn{k}{y} =y_j + \sum_{r \in N_k\backslash\{j\}}y_r\geq e^*,\qquad \forall\ k\in N_j\backslash\{i\},$$ since $y_r \geq 0$ for all $r$, and for agent $i$, $$y_i=0,\aur\scn{i}{y}=y_i+y_j=e^*.$$ 
Hence $y$ satisfies all conditions in \eqref{eq:equilibrium efforts} whereby $y$ is an equilibrium profile. Moreover, since $y_j = e^* =x_j + x_i $,
\begin{equation}\label{eq:temp6.3.2}
\scn{k}{y} = y_j + \sum_{r\in N_k\backslash \{j\}} y_r = x_i + x_j + \sum_{r\in N_k\backslash \{j\}} x_r = \scn{k}{x} + x_i, \qquad \forall\ k \in N_j \backslash \{i\}.
\end{equation} For agent $i$, since $d_i = 1$, and both $x$ and $y$ are equilibrium profiles, by Lemma \ref{lem:basic results of LCP(G)}\ref{lem:bounds on C_j}, $\scn{i}{y}=\scn{i}{x}=e^*$. Similarly since $y_j>0$ and $x_j>0$, $\scn{j}{y}=\scn{j}{x}=e^*$, due to \eqref{eq:equilibrium efforts}. 

To compare the cost of the equilibria $y$ and $x$, observe that, $$\bfe \t y = y_i + y_j +\sum_{k \neq i,j}y_k = e^*+\sum_{k \neq i,j}x_k= x_i + x_j + \sum_{k\neq i,j}x_k = \bfe \t x,$$ \ie, the total effort required by both equilibria $y$ and $x$ is the same.  Hence we have, 
\begin{align*}
\WU(y)-\WU(x) =\ & b(\scn{i}{y}) - b(\scn{i}{x}) + b(\scn{j}{y}) - b(\scn{j}{x}) + \\ &\sum_{k\in N_j\backslash\{i\}} (b(\scn{k}{y}) - b(\scn{k}{x})) + \sum_{k\notin \{j\}\cup N_j}b(\scn{k}{y}) - b(\scn{k}{x}) \\
\stackrel{(e)}{=}\ &0 + 0 + \sum_{k\in N_j\backslash\{i\}} (b(\scn{k}{y}) - b(\scn{k}{x}))  + 0 >0,
\end{align*}
where $(e)$ follows since $\scn{k}{y}=\scn{k}{x}$ for $k\notin \{j\}\cup N_j$, $\scn{i}{y}=\scn{i}{x}$ and $\scn{j}{y}=\scn{j}{x}.$ The inequality follows from \eqref{eq:temp6.3.2} and the monotonocity of $b$.  This proves the claim.
\end{proof}

The above theorem says that for a network with dependants, in a welfare maximizing equilibrium, a pendant line cannot exist as a co-specialist link. However, in general, a similar substitution of efforts amongst agents in a co-specialist pair  does not lead to higher welfare, as demonstrated in Example \ref{ex:tree}. The example shows an equilibrium where substitution of effort between co-specialist agents, which do not form a dependant-guardian pair, yields lower welfare.

\begin{examplec}\label{ex:tree}{\bf(In general, substitution among co-specialists is inefficient)}\\ 
In Figure \ref{fig:trees}, observe that in equilibrium $(B)$ the agent on the top in the vertically oriented link substitutes the effort of its co-specialist agent in equilibrium $(A)$.

Let us now calculate the difference in the welfares of both the equilibria. By Lemma \ref{lem:trees}\ref{lem:trees2}, since $(B)$ is a specialized equilibrium contained in the support of equilibrium $(A)$, the cost of both equilibria is the same, whereby the difference in the welfares of $(A)$ and $(B)$ is, 
\[\WU(A) - \WU(B) = 2\big(b(1+x) + b(2-x) - b(2) - b(1)\big) = 2x\Big(\frac{\big( b(1+x) - b(1)\big)}{x} - \frac{\big(b(2) - b(2-x)\big)}{x} \Big)>0 .\] 
The inequality above holds due to the identical lengths of intervals $[1,1+x]$ and $[2-x,2]$ and due to the concavity of the function $b$, by which the secant to $b$ between $(1,b(1))$ and $(x,b(x))$ has a higher slope than the secant between $(2-x,b(2-x))$ and $(2,b(2))$. Hence $(A)$ yields higher welfare than $(B)$.
\end{examplec}

In proving Theorem \ref{thm:well-covered}, we need another result by Plummer et. al. \cite{plummer1993well} which characterizes well-covered forest networks.

\begin{theorem}[Plummer et. al. \cite{plummer1993well} Cor. 3.2] \label{thm:plummer}
Let $G$ be a forest network over $n$ non-isolated vertices. $G$ is well-covered, if and only if, the pendant lines of $G$ form a perfect matching. As a consequence $\alpha(G)=\beta(G)=\half n$.
\end{theorem}

Note that the above theorem means that for a forest without isolated vertices to be well-covered, $(i)$ every agent is either  a dependant or a guardian and $(ii)$ every guardian has exactly one dependant. Hence the set of agents consists of $\frac{n}{2}$ disjoint guardian-dependant pairs. This can be seen in Figure \ref{fig:well_covered}. We will use this result to prove Theorem \ref{thm:well-covered}. Recall that we have to show that if a network of agents is a well-covered forest, a welfare maximizing equilibrium is necessarily specialized.

\subsection*{Proof of Theorem \ref{thm:well-covered}}

\begin{enumerate}[label=(\alph*)]
\item If there exist isolated agents, they all exert effort $e^*$, \ie, they are specialists. In the rest of the network without isolated agents, from Theorem \ref{thm:dependant_efficiency} we know that if an equilibrium in a forest network maximizes welfare, then in a dependant-guardian pair, exactly one agent is a specialist whereby the other is a free rider. However, by Theorem \ref{thm:plummer}, all the non-isolated agents in a well-covered forest network exist as guardian-dependant pairs. Hence, all agents in a welfare maximizing equilibrium are either specialists or free riders, \ie, the equilibrium is necessarily specialized. This proves part \ref{thm:well-covered welfare}.

\item Let $x$ be an equilibrium of the well-covered forest. We know from Theorem \ref{lem:NE and LCP} that $\frac{1}{e^*}x$ solves $\LCP(G)$. Since $G$ is well-covered, $\alpha(G) = \beta(G)$. Also, since $G$ is a forest $\min \{\bfe \t x \mid x \solves \LCP(G)\} = \beta(G)$ from Theorem \ref{thm:old results}\ref{thm:beta equality}. Thus we have from Theorem \ref{thm:old results}\ref{thm:wTx} that $\max \{\bfe \t x \mid x \solves \LCP(G)\}=\min \{\bfe \t x \mid x \solves \LCP(G)\}$, whereby $\bfe \t x$ is constant for all equilibria $x$, and takes the value $e^*\alpha(G)$. Moreover, from Theorem~\ref{thm:plummer}, $\alpha(G) = \half n$ for well-covered forests $G$ without isolated agents. Hence the cost of every equilibrium is $c \bfe\t x = c e^* \alpha(G) = \half ce^* n$. This proves the claim. If isolated agents are present, they exert effort $e^*$ and hence the cost incurred by each isolated agent is $ce^*$.
\end{enumerate}

\section{Conclusion}

This paper contributes to the theory of public goods provision based on the model introduced by Bramoulle and Kranton \cite{bramoulle2007public}. Our main results relate the structure of the network to its equilibria and show that specialized equilibria enjoy many attractive properties. Specifically, we showed there exists a specialized equilibrium whose welfare comes arbitrarily close to the maximum equilibrium welfare as the concavity of the benefit function approaches unity, while for forest networks this holds true even as the concavity approaches zero. More generally, for any network, there exists a specialized equilibrium which maximizes total weighted equilibrium effort exerted by the agents. On forest networks they minimize the total cost amongst all equilibria. This makes the case that the specialized equilibrium may be considered as a refinement of the Nash equilibrium of a public goods game.

We further study the Nash equilibria of the public goods game in networks with dependants, \ie, agents adjacent to exactly one other agent, and show that they have a special structure, due to which there always exists a free riding agent in every equilibrium. 
We also study a class of networks called well-covered forests which possess the property that any welfare maximizing equilibrium is necessarily specialized, and that every equilibrium in a public goods game over such a network incurs the same cost.

\section*{Appendix}
\begin{proof}[of Lemma \ref{lem:basic results of LCP(G)}] 
\begin{enumerate}[label=(\alph*)]

\item Recall that a set $S\subseteq V$ is said to be dominating if any vertex not in $S$ has at least one neighbour in $S$. From \eqref{eq:equilibrium efforts}, we know that for an equilibrium $x$, we have $\scn{i}{x} \geq e^*$ for all agents $i$. An agent $i$ who is not a supporting agent is a free rider, \ie $x_i = 0$, whereby $\scn{i}{x}:=x_i + \sum_{j \in N_i}x_j = \sum_{j \in N_i}x_j \geq e^*$, where $N_i$ is the set of agents adjacent to $i$. Since all $x_j \geq 0$, there is at least one neighbour of $i$ who exerts positive effort, \ie, who is a supporting agent. This proves the claim.

\item Let $x$ be an equilibrium effort profile and let agent $i^*$ be a free rider, \ie, $x_{i^*} = 0$. At equilibrium for every agent $i$, $x_i \geq 0$, $\scn{i}{x} \geq e^*$, and at least one inequality is tight. Observe that since $i^*$ is a free rider at equilibrium, it doesn't contribute to the benefit of any agent in the network and hence the effort of the closed neighbourhood $\scn{i}{x}$ for all agents other than $i^*$ remains the same if $i^*$ left the network. Thus if $y:=x_{(-i^*)}$, then for all $i \neq i^*$, $y_i=x_i$ and $\scn{i}{y} = \scn{i}{x}$. Thus $y_i \geq 0$, $\scn{i}{y} \geq e^*$ and at least one inequality is tight whereby $y$ is an equilibrium effort profile, which proves the claim.

\item If agent $i$ is a specialist in an equilibrium with effort profile $x$, \ie, $x_i = e^*$, and if it's neighbour $j$ exerts an effort $x_j > 0$, then $\scn{i}{x}\geq x_i + x_j>e^*$, which contradicts the third condition in \eqref{eq:equilibrium efforts}. Hence $x_j = 0$ for all neighbours $j$ of $i$, whereby any two agents who are specialists in an equilibrium cannot be adjacent. Hence the specialists form an independent set. Note that this need not be a maximal independent set.

\item If agent $i$ and $j$ form a co-specialist pair in equilibrium with effort profile $x$, \ie, $x_i  + x_j = e^*$. If a neighbour $k$ of agent $i$ exerts positive effort $x_k >0$, then $\scn{i}{x} \geq x_i + x_j + x_k > e^*$, which contradicts \eqref{eq:equilibrium efforts}.  Hence $x_k=0$ for all neighbours of $i$ and $j$ which form a co-specialist pair. This means that for any two pairs of co-specialists in an equilibrium, there is no agent common to the both pairs, whereby the co-specialist links form a matching of the network.

\item Assume $d_i>0$ for an agent $i$. If $i$ is a free rider, observe that $\scn{i}{x} := x_i + \sum_{j\in N_i} x_j  \leq |N_i|e^* = d_i e^*$, where equality is attained only if all neighbours $j \in N_i$ are specialists, which proves the third inequality. On the other hand if $i$ is not a free rider,then by \eqref{eq:equilibrium efforts} $\scn{i}{x} = e^*$, whereby the second inequality in \eqref{lem:inequalities in efforts} and its equality condition are proved. Finally, the first inequality in \eqref{lem:inequalities in efforts} is tight if $i$ is a specialist, which follows directly from the definition. However in this case, the second inequality is also tight, since a specialist is not a free rider. If $d_i=0$, then clearly the agent's effort in any equilibrium $x$ is $e^*$ and $\scn{i}{x}=e^*.$
\end{enumerate}
\end{proof}

\section*{Acknowledgments}
The research of the authors was supported by a grant by the Science and Engineering Research Board, Department of Science and Technology, Government of India.

\bibliographystyle{plainini}
\bibliography{ref}
\end{document}